\documentclass[12pt]{iopart}

\usepackage{iopams}  

\expandafter\let\csname equation*\endcsname\relax
\expandafter\let\csname endequation*\endcsname\relax

\usepackage{graphicx,subfigure}
\usepackage{amssymb,amsmath,amsfonts,palatino,amsthm}
\usepackage{bbm}
\usepackage{bbold}
\usepackage{booktabs}
\usepackage{braket}
\usepackage{enumitem}
\usepackage{hhline}
\usepackage{mathtools}
\usepackage[inkscapeformat=png]{svg}

\newtheorem{theorem}{Theorem}
\newtheorem{corollary}{Corollary}[theorem]
\theoremstyle{plain}

\newtheorem{lemma}[theorem]{Lemma}
\newtheorem{proposition}[theorem]{Proposition}
\theoremstyle{plain}

\DeclarePairedDelimiter{\ceil}{\lceil}{\rceil}
\DeclarePairedDelimiter{\floor}{\lfloor}{\rfloor}

\newcommand{\pdv}[3][]{\frac{\partial^{#1} #2}{\partial #3^{#1}}}

\usepackage[backend=biber,isbn=false,url=false,sorting=none,style=numeric-comp]{biblatex}
\usepackage[pdftex,pagebackref=false]{hyperref} 
\usepackage{orcidlink}
\hypersetup{
    plainpages=false,       
    unicode=false,          
    pdftoolbar=true,        
    pdfmenubar=true,        
    pdffitwindow=false,     
    pdfstartview={FitH},    
    pdftitle={Effects of detuning on \texorpdfstring{$\mathcal{PT}$}{PT}-symmetric, tridiagonal, tight-binding models},    
    pdfauthor={Jacob L. Barnett},    
    bookmarksnumbered = true, 
    linktoc = none,
    pdfnewwindow=true,      
    colorlinks=true,        
    linkcolor=blue,         
    citecolor=green,        
    filecolor=magenta,      
    urlcolor=cyan           
}

\usepackage{cleveref}

\makeatletter
\long\def\@makecaption#1#2{%
  \vskip 10\p@
  \begin{center}
    \begin{minipage}{0.9\textwidth}
      \textbf{#1.} #2
    \end{minipage}
  \end{center}
  \vskip 0\p@
}
\makeatother

\begin{document}
\title{\texorpdfstring{Effects of detuning on $\mathcal{PT}$-symmetric,  tridiagonal,  tight-binding models}{Effects of detuning on PT-symmetric,  tridiagonal,  tight-binding models}}

\author{Jacob L. Barnett$^{1}$\orcidlink{0000-0002-7101-4229} and Yogesh N. Joglekar$^2$\orcidlink{https://orcid.org/0000-0002-3222-1436}}
\address{$^1$Basque Center for Applied Mathematics, Bilbao, Bizkaia, 48009 Spain}
\address{$^2$ Department of Physics, Portland State University, Portland, Oregon 97207 United States}
\ead{jbarnett@bcamath.org}
\ead{yogesh@pdx.edu}

\begin{abstract}
Non-Hermitian,  tight-binding $\mathcal{PT}$-symmetric models are extensively studied in the literature.  Here, we investigate two forms of non-Hermitian Hamiltonians to study $\mathcal{PT}$-symmetry breaking and features of surfaces of exceptional points (EPs).  They include one-dimensional chains with 
2-periodic tunnelling strengths and a pair of balanced gain and loss potentials at parity-symmetric sites. By considering several examples,  we demonstrate that algebraic EP curves of a given order generically have cusp-points where the order of the EP increases by one. We count the number of edge states when the gain and loss sites are at the edges of an even chain and show that this number is dictated by the $\mathcal{PT}$ and topological phases. When the gain and loss sites are nearest neighbour at the center of an even chain, we construct a positive-definite metric operator for the entire domain of quasi-Hermiticity and show the entire spectrum becomes non-real as the gain is tuned above an EP. Taken together, our results provide a detailed understanding of detuned tight-binding models with a pair of gain-loss potentials.  
\end{abstract}

\section{Introduction}
\label{sec:intro}

\subsection{Background}



Physical systems exhibiting a superposition principle are often mathematically modeled as linear dynamical systems. For example, in quantum physics, time evolution is dictated by the Schr\"odinger equation and is generated by an operator acting on a Hilbert space referred to as the Hamiltonian. In quantum theory, 
the Hamiltonian is required to be Hermitian so that time-evolution satisfies the physical requirement of unitarity.



Unitary time evolution is only a requirement for \textit{isolated} quantum systems.  When a quantum system interacts with an environment,  the resulting dynamics for the reduced density matrix of the system is typically decoherence inducing.  Under mild conditions, such as a Markovian bath,  this dynamics is described by a completely positive trace preserving (CPTP) map that is generated by the Lindblad equation~\cite{Gorini1976,Lindblad1976}. Furthermore, non-Hermitian Hamiltonians 
emerge as effective descriptions of classical systems with gain and loss~\cite{Joglekar2013}, notably in photonics~\cite{ElGanainy2007,Guo2009,Rter2010,ElGanainy2018}. 
In the truly quantum domain,  it has been shown that they emerge from Lindblad equation through post-selection where trajectories with quantum jumps are eliminated~\cite{Naghiloo2019}.   Additional examples of quantum phenomena modelled by non-Hermitian Hamiltonians include radioactive decay in nuclear systems \cite{Siegert1939,Feshbach1958,Feshbach1962} and renormalization in quantum field theories \cite{quantumRG,LeeModel,kallen1955mathematical,LeeModelPT}. 


A noteworthy class of non-Hermitian Hamiltonians are those with an antilinear symmetry. These Hamiltonians commute with an invertible antilinear operator, $\Theta$, which provides a generalized notion of time-reversal symmetry \cite{Wigner1932}.  The eigenvalues of these Hamiltonians are real or occur in complex-conjugate pairs because $\Theta$ maps an eigenpair $(\lambda, v)$ to the eigenpair $(\lambda^*,  \Theta v)$.  This feature - the pairing of complex-conjugate eigenvalues - reflects a notion of balanced gain and loss in the system.
Additionally,  if a Hamiltonian has $\textit{unbroken}$ antilinear symmetry,  meaning all of its eigenspaces are $\Theta$-invariant, the Hamiltonian's spectrum is real \cite{bender1999pt}. 
For historical reasons, the antilinear symmetry is often factorized as a product of a linear part and an antilinear part that acts as complex conjugation in a suitable orthonormal basis; these parts are called \textit{parity} $(\mathcal{P})$ and \textit{time-reversal} $(\mathcal{T})$ symmetries, respectively and $\Theta = \mathcal{PT}$ \cite{bender1998real}. 

$\mathcal{PT}$-symmetry breaks as a parameter is deformed past an \textit{exceptional point}, which intuitively occurs when eigenspaces coalesce. More precisely, consider a Hamiltonian, $H$, dependent on some parameters in a measure space. If $H$ is diagonalizable almost everywhere, then the parameters that define non-diagonalizable Hamiltonians are called exceptional points (EPs) \cite{Kato1995}. 
A $k$-th \textit{order} EP, called an EP$k$ for brevity, is a spot where there are $k$ coalescent eigenvectors \cite{Graefe2008,Heiss2008}.




In finite-dimensional Hilbert spaces, an antilinear symmetry for $H$ exists if and only if $H$ is \textit{pseudo-Hermitian} \cite{mostafazadeh2002pseudo3,Scolarici2003,siegl2008quasi,Siegl2009}. 
Pseudo-Hermitian operators are those such that there exists an invertible Hermitian \textit{intertwiner}, $\eta=\eta^\dagger$, satisfying \cite{Heisenberg1957}
\begin{equation}
\eta H = H^\dag \eta. \label{Dieudonne}
\end{equation}
If a positive-definite intertwiner exists, we refer to it as a \textit{metric operator}, and $H$ is called \textit{quasi-Hermitian} \cite{dieudonne,QuasiHerm92}. A finite-dimensional matrix is diagonalizable with a real spectrum if and only if it is quasi-Hermitian \cite{Drazin1962,Taussky1960,mostafazadeh2002pseudo2,mostafazadeh2008metric}. 
Furthermore, a quasi-Hermitian operator is a symmetric operator with respect to an alternative inner product defined by the metric. 
Thus, a fundamental extension of quantum theory with quasi-Hermitian observables is well defined~\cite{QuasiHerm92} when the metric has a bounded inverse \cite{siegl2012metric,MyErratum}.  Alternatively,  if non-Hermitian Hamiltonians are considered an effective description, where the loss of unitarity is not prohibited, physical predictions can be derived using the original (canonical) inner product and the intertwiners take the role of time invariants \cite{Pauli1943DiracFieldQuantization,bian2019time}.  



An extensive body of literature has been devoted to characterizing spectra, exceptional points, and $\mathcal{PT}$-symmetry breaking for finite-dimensional matrices. A typical design recipe for generating interesting $\mathcal{PT}$-symmetric models is to perturb a well-understood Hermitian Hamiltonian by a purely imaginary potential 
\cite{MyFirstPaper,Schomerus2013,Babbey,Jin2017,Klett2017,
ortega2019mathcal,jin2009solutions,Turker2019,
modak2021eigenstate,Xu2020,Zhu2014,Lieu2018,
Mochizuki2020,InfiniteLattice,PTRing,LiLiZhangGong2020,
Dangel2018,Lang2018,Harter2018Fragile,MartinezMartinez2024,BarnettElGanainy2026}. Here, we extend a number of intriguing results known for $\mathcal{PT}$-symmetric models to scenarios where the potential has nonzero real and imaginary parts. 
Including a real-valued potential that we refer to as \textit{detuning} 
enlarges parameter space and, thus, defines algebraic curves of exceptional points with intriguing geometry. Exceptional curves are a recent subject of interest  \cite{Tang2020ExceptionalNexus,EPCusps,Ding2022,Carlstrom2018,Carlstrom2019}. 
Simplifications in analysis occur in the zero-detuning case, where the Hamiltonian can exhibit  particle-hole symmetry \cite{Chiu2016}; this is likely related to the comparatively larger body of literature devoted to cases without detuning. 



To generate results quickly at the potential sacrifice of rigor, the analysis of tight-binding models is often simplified by performing numerical experiments, 
 imposing strong symmetry constraints such as translation invariance, 
 or making approximations that become justifiable in the thermodynamic limit. 
 Here, by using algebraic methods involving orthogonal polynomials, we provide analytic results on finite-size properties that, in most cases, is rigorously non-approximative. A similar body of work on topological metamaterials utilizing orthogonal polynomials is summarized in \cite{AMMARI2025100103}. 
 

\section{Overview}


We study two families of perturbed tridiagonal matrices introduced in \cref{sec:model}, with results summarized in \cref{sec:nearest-summary,subSSH} and established in \cref{sec:openChain,sec:SSH}, respectively. 

EP3s were located at cusp singularities of EP curves in a non-Hermitian Lieb lattice model \cite{EPCusps} and in a ternary acoustic cavity system \cite{Tang2020ExceptionalNexus}. In \cref{sec:SSH}, we report another example of this phenomenon in a non-Hermitian variant of the Su-Schrieffer-Heeger (SSH) chain \cite{SSH}. Subsequently, in \cref{sec:singular-general}, we argue that when the EP set is an algebraic curve, then EP3s are typically located at the cusp-points of this curve.


\subsection{Model and Notation} \label{sec:model}

This report considers models of dynamics on the two kinds of one-dimensional $n$-site connected graphs, cycles and chains. We take the state space to be the Hilbert space $\mathbb{C}^n$ and define the actions of parity and time-reversal operators to be
\begin{align}
\mathcal{P}_n e_k &= e_{\overline{k}} &\quad
\mathcal{T} e_{k} &= e_{k}, \label{T}
\end{align}
where $(e_k)_j := \delta^k_{j}$ enumerates the canonical basis vectors of $\mathbb{C}^n$ for $j, k \in \{1, \dots, n\}$, $\delta$ is the Kronecker delta, and $\overline{k} := n+1-k$. Complex matrices, $A \in \mathbb{C}^{n \times n}$, are $\mathcal{PT}$-symmetric in this sense if and only if $A_{jk}=A^{*}_{\overline{j}\overline{k}}$ and are called \textit{centrohermitian} \cite{Lee1980}. We denote the spectrum of $A$ by $\sigma(A)$ and its monic characteristic polynomial by
\begin{align}
    p_A(\lambda) := \det(\lambda \mathbb{1}_n - A),
\end{align}
where $\mathbb{1}_n \in \mathbb{C}^{n \times n}$ denotes the identity matrix. The real and imaginary parts of a complex number $z$ are denoted by $\Re z$ and $\Im z$, respectively.

The Hamiltonians investigated in this paper are special cases of transpose-symmetric tridiagonal matrices over $\mathbb{C}^n$ with perturbed corners,
\begin{align}
H_{ij} &= z_i \delta^i_{j} + t_j \delta^i_{j+1} + t_i \delta^j_{i+1} + \delta^i_{\bar{j}}(t_L \delta^i_{1} + t_R \delta^i_{n}) 
\nonumber\\
&= \begin{pmatrix} 
z_1 & t_1 & 0 & \dots &0 &  t_L \\
t_1 & z_2 & t_2 & \ddots & \ddots & 0 \\
0 & t_2 & z_3 & \ddots & \ddots & \vdots \\
\vdots & \ddots & \ddots & \ddots & t_{n-2} & 0 \\
0 & \ddots & \ddots & t_{n-2} & z_{n-1} &t_{n-1} \\
t_R & 0 & \dots & 0 & t_{n-1}  & z_n
\end{pmatrix}_{ij},
\label{TriDiag}
\end{align}
where $z_j, t_L, t_R \in \mathbb{C}$ and $t_j\in\mathbb{R}$. Since $H$ is a weighted adjacency matrix of a cycle, we interpret the dynamics it generates as \textit{local} to a one-dimensional cycle. 
$\mathcal{PT}$-symmetric instances of $H$  have been well explored 
\cite{Korff2008,Xu2020,jin2009solutions,
InfiniteLattice,Babbey,JoglekarSaxena,PTRing,MyFirstPaper,
ortega2019mathcal,farImpurityMetric,guo2016solutions,Joglekar2010,Znojil2009,ZnojilGeneralized,Ruzicka2015,Zhu2014,Klett2017,Jin2017,Lieu2018,Yao2018,Turker2019,Mochizuki2020,Harter2018Fragile} and an incomplete list of spectral results pertaining to $H$ may be found in \cite{rutherford1948xxv,elliott1953characteristic,losonczi1992eigenvalues,
yueh2005eigenvalues,da2007characteristic,YUEH2008,Willms2008,Chang2009,Joglekar2010,
daFonseca2015,KILI2016,Chu2010,Chu2019,Alazemi2021,BagarelloHatano2021}.

\Cref{chains} illustrates the two families of Hamiltonians studied here: an even open chain with balanced gain and loss on the central two sites, and an SSH chain perturbed at its edges. 




\begin{figure}[htp!] 
\centering
\subfigure[]{
\includegraphics[width = .95\textwidth]{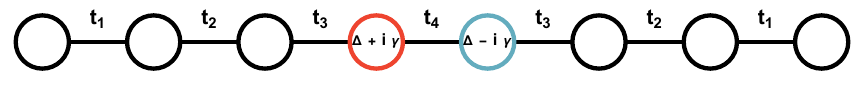}
}
\subfigure[]{
\includegraphics[width = .95\textwidth]{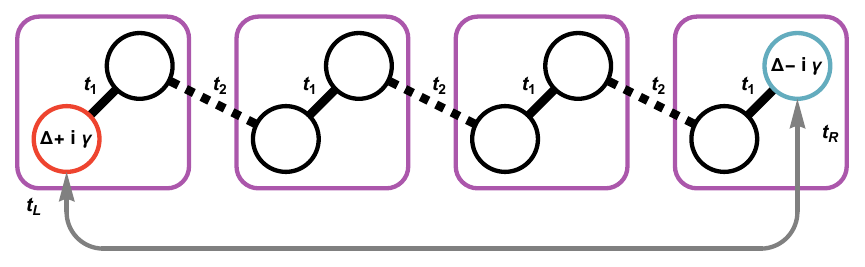}
}
\caption{Graphical representation of the two types of Hamiltonians studied in this report for $n = 8$ sites. Case (a) is an open chain with arbitrary $\mathcal{PT}$-symmetric tunnelling strengths and non-Hermitian nearest neighbour defect potentials. Case (b) perturbs the Su-Schrieffer-Heeger (SSH) model \cite{SSH} by a non-Hermitian matrix at the edges of the chain. The SSH model consists of $\floor{n/2}$ connected copies of two-site cells, depicted in purple, with intracell and intercell tunnellings $t_1$ and $t_2$, respectively.} \label{chains}
\end{figure}

We now state additional assumptions that we make on the model parameters. 
\begin{itemize}
    \item We select a pair of mirror-symmetric sites called \textit{defects}, $(m, \overline{m})$ with $m \in \{1, \dots, n \}$ and $m \neq \overline{m}$.
    Away from the defects, we assume the diagonal elements of $H$ (\textit{potentials}) are real and parity-symmetric, so $z_j = z_{\overline{j}} \in \mathbb{R} \, \, \forall j \notin \{m ,\bar{m} \}.$ The parameters $\Re(z_j)$ are called \textit{detuning}. To simplify select equations, we parametrize the defect potentials as $z_m=\Delta+\mathfrak{i} \gamma$ and $z_{\overline{m}}=\Delta-\mathfrak{i}\gamma$.

    \item Imposing \textit{open boundary conditions} means setting $t_L = t_R = 0$; we only consider open chains in \cref{sec:openChain} and \cref{sec:critical-EPs,sec:edges,sec:EP3Contour}. 
    \item As we are primarily interested in the $\mathcal{PT}$-symmetric case, we demand $z_m = z_{\overline{m}}^*$ (except in \cref{esystemSection}) and $t_j = t_{n-j}$ (except in \ref{app:odd}). When $H$ is $\mathcal{PT}-$symmetric, $\Delta, \gamma \in \mathbb{R}$ holds. 
    Standard arguments summarized in \ref{app:symmetry} let us take $\gamma \geq 0$ (when $z_m = z_{\overline{m}}^*$) and $H$ to be \textit{irreducible}, meaning $t_j > 0$,  without loss of generality. 
    \item From \cref{sec:exact-spectra} onwards, we assume the detuning is \textit{defect-localized}, meaning 
    detuning is nonzero only at the defects, i.e. $z_j \neq 0 \, \Rightarrow j \in \{m, \bar{m}\}$.
\end{itemize}

The $\mathcal{PT}$-\textit{threshold}, $\gamma_\mathcal{PT}$ is the infimum of the positive gain--loss strengths for which the spectrum is not
entirely real, with other parameters considered fixed.

The eigenvalues of irreducible tridiagonal matrices always have geometric multiplicity equal to one \cite[Thm. 4.2]{elliott1953characteristic}; \ref{app:geom-mult} extends this known result to $H$ of \cref{TriDiag} when $t_L \neq t_R$.


\subsection{Results on nearest neighbour defects} \label{sec:nearest-summary}

In \cref{sec:openChain}, we consider open chains with nearest-neighbour defects. 
In \cref{thm:nearest-neighbor}, we show how the $\mathcal{PT}$-threshold and quasi-Hermiticity are determined by a single real parameter, $\gamma/t_m$. This feature is rather unusual since the model and spectrum also depend upon $\mathcal{O}(n)$ other parameters. The system undergoes maximal symmetry breaking~\cite{MyFirstPaper} when $\gamma$ is tuned above the threshold $\gamma_{\mathrm{PT}} = t_m$, where all of the eigenvalues coalesce in pairs at an EP2 and subsequently leave the real axis. 
This is atypical in the sense that no EP3s or cusp points exist.

We construct a one-parameter family of intertwiners in \cref{homomorphismMetric}, which includes positive-definite metrics in the $\mathcal{PT}$-unbroken domain. Using these metrics, we determine a corresponding $\mathcal{C}$-symmetry \cite{bender2002complex,CorrectCPT} in \cref{generalC} and similar local Hermitian Hamiltonians in \cref{equivHermHam}. The latter confirms a conjecture of \cite{modak2021eigenstate} and contrasts with the generic case where similar Hermitian Hamiltonians are nonlocal \cite{Korff2008}.

When the tunnelling is uniform, $t_i = t$, we compute the spectrum of $H$ in closed-form for some special cases of defect potentials, summarized in \cref{table}. Since these eigenvalues have geometric multiplicity equal to one \cite[Thm. 4.2]{elliott1953characteristic}, the cases in \cref{table} where $H$ has less than $n$ distinct eigenvalues are exceptional points.

For defects located inside a uniform chain, instead of at its end-points, we find that a subset of the spectrum is independent of the defect potential, $z_m$, whenever $m$ shares a nontrivial factor with $n+1$; 
Exceptional points occur when these eigenvalues are multiple roots of the characteristic polynomial. In general, these exceptional points do not coincide with the $\mathcal{PT}$-symmetry breaking threshold, and the spectrum is generically complex in the vicinity of these exceptional points. These eigenvalues are a strict subset of those identified at zero detuning in \cite{ortega2019mathcal}.

\subsection{Results on a perturbed SSH chain}
\label{subSSH}

Section~\ref{sec:SSH} treats a non-Hermitian perturbation of the SSH model with complex defect potentials at the chain's edges ($m = 1$) and defect-localized detuning.  The Hermitian SSH Hamiltonian is a tridiagonal matrix with 2-periodic tunnelling strengths, given by $t_1,t_2$ respectively. We allow for non-open boundary conditions, where $t_L$ and $t_R$ satisfy the relaxed constraints of \cref{eqn:PBC}. 
The spectral effects of a pair of non-Hermitian defect potentials have been studied extensively both for critical SSH chains, where the model reduces to a uniform tight-binding chain 
\cite{Korff2008,jin2009solutions,Babbey,farImpurityMetric,ortega2019mathcal,Willms2008,MyFirstPaper}, and for dimerized, noncritical SSH chains \cite{Zhu2014,Klett2017,Ruzicka2015,Yuce2015,Jin2017,LiLiZhangGong2020,Xu2020,MyFirstPaper}. Alternative non-Hermitian perturbations of the SSH chain have also been well-explored \cite{Ruzicka2015,Jin2017,Lieu2018,Yao2018,Turker2019,Mochizuki2020}.

The unperturbed Hermitian SSH model exhibits two distinct topological phases; the defining characteristic of a topological phase is that a continuous change in parameters may only induce a quantum phase transition when the band gap closes \cite{Hasan}. Due to the bulk-boundary correspondence \cite{Chen2020,Prodan2016,Rhim2018,Pershoguba2012}, these phases, called trivial ($t_2<t_1$) or non-trivial ($t_2>t_1$), are characterized by the non-existence or existence of \textit{edge states} of open chains, where an edge state is an eigenstate whose component magnitudes are peaked at the edges of the chain and decay exponentially as the site index is moved inwards.

This topological distinction has direct consequences for $\mathcal{PT}$-symmetry breaking in the presence of balanced gain and loss: in the thermodynamic limit $(n \to \infty)$, the edge states are localized on one of the ends and therefore couple predominantly to only one of the defect potentials, rendering the $\mathcal{PT}$-unbroken phase fragile \cite{Znojil2004} in the sense that non-real eigenvalues corresponding to edge states exist when $\gamma > 0$ \cite{Zhu2014,Klett2017}. Conversely, the $\mathcal{PT}$-unbroken phase is robust (not fragile) when the Hermitian limit does not contain edge states.

In \cref{sec:evalue-inclusion,sec:edges}, we generalize this picture to include detuned cases by characterizing the $\mathcal{PT}$-phase diagram and counting the numbers of real eigenvalues and edge states. 
Our results are pictorially summarized in the phase diagrams of \cref{edgeStateCounter,fig:PTBreaking}. 
\Cref{inclusionTheorem} gives sharp eigenvalue inclusion criteria guaranteeing real eigenvalues in intervals related to the band structure; $n-4$ distinct eigenvalues are real and reside in the bands when $n$ is even. The entire spectrum is real whenever a set of parametric constraints listed in \cref{inclusionTheorem} hold;  in particular, for open non-Hermitian chains, this includes a cylinder and some arcs,
\begin{align}
&\Delta^2+\gamma^2<t_2^2\leq t_1^2,
\quad\text{or}\quad
\Delta^2+\gamma^2=t_2^2,\qquad \gamma^2\leq t_1^2.
\end{align}
This condition is necessary and sufficient for a real spectrum in the thermodynamic limit.
Conversely, we 
show how the existence of edge states is directly linked with fragile $\mathcal{PT}$-symmetry breaking.

For uniform chains, we parametrize the exceptional contour by the defective eigenvalue in proposition~\ref{prop:EP2-parameter} and express the contour algebraically using resultants in \cref{EPSurface}. Subsequently, we deduce an asymptotic relation for the $\mathcal{PT}$-threshold for large detuning, $\gamma_{\mathcal{PT}}/t\sim(\Delta/t)^{-(n-2)}$. Finally, the spectral response to perturbations is examined in \cref{sec:EP3Contour}, where divergences provide evidence that the cusp singularities of this contour are EP3s. 

To place our work in a broader context, we mention that another well-studied non-Hermitian extension of the SSH model places
alternating gain and loss on its two sublattices
\cite{Schomerus2013,Lieu2018,Dangel2018,Lang2018,Slootman2024,Klett2017,Harter2018Fragile}.
Such models can be regarded as nontrivial square roots of energy-shifted
supersymmetric Hamiltonians
\cite{MartinezMartinez2024,BarnettElGanainy2026}. Consequently, in \ref{app:alternating}, we utilize this observation in combination with our eigenvalue inclusion results to place bounds on the $\mathcal{PT}$-threshold that hold in finite-size regimes.


\section{Open chain with nearest neighbour defects} \label{sec:openChain}

This section presents analytical results for a non-uniform open chain with an even number of sites and nearest neighbour defects positioned at the center ($n = 2m$). 

\subsection{\texorpdfstring{$\mathcal{PT}$-symmetric phases}{PT-symmetric phases}}

We start by characterizing when the Hamiltonian is diagonalizable and how many of its eigenvalues are real. 

\begin{theorem} \label{thm:nearest-neighbor}
    Consider an open chain with nearest-neighbour defects $(n = 2m, m \in \mathbb{N}_+)$. Then,
    \begin{enumerate}
        \item \label{thm:pt:metric} When $\gamma \neq 0$, a matrix is an intertwiner for $H$ and a sum of a diagonal and anti-diagonal matrix if and only if it is a linear multiple of $\mathcal{P}$ or of any the block matrices defined by 
\begin{align}
\eta^{(m)}(Z) &= \begin{pmatrix}
\mathbb{1}_m & \frac{Z^*}{t_m} \mathcal{P}_m \\
\frac{Z}{t_m} \mathcal{P}_m & \mathbb{1}_m
\end{pmatrix}, \label{homomorphismMetric} 
\end{align}
where $Z \in \mathbb{C}$ has an arbitrary real part and $\Im Z = \gamma$. 
        \item \label{thm:pt:noreal} $H$ has no real eigenvalues if and only if $\gamma > t_m$.
        \item \label{thm:pt:ep} $H$ is non-diagonalizable with a real spectrum if and only if $\gamma = t_m$. In this case,       the complete set of eigenvectors of $H$ has $m$ orthogonal elements. 
        \item \label{thm:pt:qh}$H$ is diagonalizable with a real spectrum if and only if $\gamma < t_m$. 
    \end{enumerate}
\end{theorem}
\begin{proof}
Since \cref{thm:pt:ep,thm:pt:noreal,thm:pt:qh} exhaust all possible values for $\gamma$, proving their "if" directions is sufficient to establish their converses. This is carried out below.
\begin{enumerate}
    \item We proceed by induction on $m$. 
    When $m=1$, the most general intertwiner is either a linear multiple of \cite{mosta2003equivalence,wang20132}
\begin{equation}
\eta^{(1)} = \begin{pmatrix}
1 & Z^*/t_m \\
Z/t_m & 1
\end{pmatrix}.
\end{equation}
or a multiple of $\mathcal{P}$.
Let $m_0 \in \mathbb{N}$ and suppose every intertwiner for the matrix $H$ when $m = m_0$ that is a sum of a diagonal and an antidiagonal part is of the form $\eta^{(m_0)}$ given in \cref{homomorphismMetric} or a linear multiple of $\mathcal{P}$. Then, the pseudo-Hermiticity constraint on $H$ for $m = m_0$, \cref{Dieudonne}, is equivalent to the following identity  for $m=m_0+1$, 
\begin{equation}
\sum^{2 m_0+1 }_{j=2} \eta^{(m_0 + 1)}_{ij} H_{jk} = \sum^{2 m_0+1}_{j=2} H^{\dag}_{ij} \eta^{(m_0 + 1)}_{jk}, \,\,\,\,\, \forall \,\, 1 < i,k < n.
\end{equation} 
Thus, $\eta^{(m_0+1)}$ is a sum of a diagonal and an antidiagonal matrix and satisfies the intertwining relation of \cref{Dieudonne} if and only if
\begin{align}
\begin{pmatrix}
\eta^{(m_0+1)}_{1,1} & \eta^{(m_0+1)}_{1,n} \\
\eta^{(m_0+1)}_{n,1} & \eta^{(m_0+1)}_{n,n}
\end{pmatrix} &= 
\begin{pmatrix}
\eta^{(m_0+1)}_{2,2} & \eta^{(m_0+1)}_{2,n-1} \\
\eta^{(m_0+1)}_{n-1,2} & \eta^{(m_0+1)}_{n-1,n-1}
\end{pmatrix},
\end{align}
which implies $\eta^{(m_0+1)}$ be a multiple of $\mathcal{P}$ or one of the matrices in the form of  \cref{homomorphismMetric}. 
\item We start by expressing the characteristic polynomial, $p_H$, as a sum of squares. Let $H_j$ denote the leading principal submatrix of $H$ of order $j \in \mathbb{N}$, i.e. the $j \times j$ matrix formed by taking the matrix elements from the first $j$ rows and columns of $H$. Then, owing to the multilinearity property of determinants, we find 
\begin{align}
    p_H(\lambda) = (\gamma^2-t_m^2) p_{H_{m-1}}^2(\lambda) + \left[(\lambda - \Delta) p_{H_{m-1}}(\lambda) - t_{m-1}^2 p_{H_{m-2}}(\lambda) \right]^2, \label{eqn:poly-sum-of-squares}
\end{align}
where we define $p_{H_0} = 1, p_{H_{-1}} = 0$. 

Suppose $\lambda \in \mathbb{R} \cap \sigma(H)$ is a real eigenvalue of $H$. Since the matrix elements of $H_j$ are real-valued for all $j \leq m-1$, $p_H(\lambda) = 0$ is a sum of two non-negative numbers appearing in the right-hand side of \cref{eqn:poly-sum-of-squares} whenever $\gamma > t_m$. Thus, both of these numbers must be zero, implying 
\begin{align} 
p_{H_{m-1}}(\lambda) = p_{H_{m-2}}(\lambda) = 0 &\quad& \forall \lambda \in \mathbb{R} \cap \sigma(H) \label{eqn:vanishing-continuant}.
\end{align} Note the characteristic polynomials of the leading principal submatrices satisfy the continuant recurrence
\begin{align}
    p_{H_j}(\lambda) = (\lambda - z_j) p_{H_{j-1}}(\lambda) - t_{j-1}^2 p_{H_{j-2}}(\lambda) &\quad& \forall j \in \{1, \dots, n\}. \label{eqn:continuant-recurrence}
\end{align}
\Cref{eqn:vanishing-continuant,eqn:continuant-recurrence} imply $p_{H_j}(\lambda) = 0$ for all $j \in \{0, \dots, m\}$ by virtue of recursion. Since this contradicts the definition $p_{H_0}(\lambda) = 1$, $\lambda \in \mathbb{R} \cap \sigma(H)$ cannot exist.

\item First, we prove that $H$ has at most $m$ linearly independent eigenvectors when $\gamma = t_m$. Since \cref{eqn:poly-sum-of-squares} implies $p_H$ is the square of a degree-$m$ monic polynomial in this case, each eigenvalue of $H$ has an algebraic multiplicity of at least two. Since these eigenvalues have a geometric multiplicity of one \cite[Thm. 4.2]{elliott1953characteristic}, there are at most $m$ linearly independent eigenvectors.  

The existence of a positive semi-definite intertwiner of $H$, namely 
\begin{align}
    G = \eta(i \gamma),
\end{align}
may be used to show that $H$ has at least $m$ eigenvectors when $\gamma = t_m$ \cite[Thm. 1]{Drazin1962}. We outline this process below. Consider the action of $H$ on $\ker G$.  An orthonormal basis of $\ker G$ is
\begin{align}
\ker G &= \text{span} \{\tilde{e}_j\,|\,j \in \{1, \dots, m\} \}, \\
\tilde{e}_j &= \frac{\mathfrak{i} \gamma}{\sqrt{2} t_m} e_j + \frac{1}{\sqrt{2}} e_{\bar{j}}.
\end{align}
Then,
\begin{align}
H \tilde{e}_j &= \Re(z_j) \tilde{e}_j +  \begin{cases}
t_{1} \tilde{e}_2 & \text{ if } j = 1\\
t_{j-1} \tilde{e}_{j-1} + t_{j} \tilde{e}_{j+1} & \text{ if } j \in \{2, \dots, m-1\} \\
t_{m-1} \tilde{e}_{m-1} & \text{ if } j = m.
\end{cases} \label{tildeH}
\end{align}
Thus, $\ker G$ is an invariant subspace of $H$ and the restriction $H|_{\ker G}$ is Hermitian, so $H|_{\ker G}$ has $m$ orthogonal eigenvectors corresponding to real eigenvalues.  These $m$ eigenvectors are also eigenvectors of $H$ with real eigenvalues. 
\item 
Note $\eta^{(m)}(Z)$ is positive-definite if and only if $|Z|<t_m$. 
Thus, when $\gamma < t_m$, $H$ is quasi-Hermitian with $G$ as a metric operator, implying $H$ is diagonalizable with a real spectrum in this case \cite{Taussky1960,Drazin1962}. 
\end{enumerate}
\end{proof}

The preceding theorem implies $\gamma = t_m$ is a second-order EP and that the $\mathcal{PT}$-unbroken domain is $\gamma < t_m$. We remark that \cref{thm:pt:metric} was observed in \cite{Barnett_2021} and that the zero-detuning case of \cref{thm:pt:noreal} was established in \cite{MyFirstPaper} by employing a different argument that is reliant upon the $\mathcal{PT}$-symmetry of $H$. By generalizing this argument to cases with nonzero detuning in \ref{app:maximal-breaking}, we establish an alternative proof of \cref{thm:pt:noreal}. 

We remind the reader that a general tridiagonal matrix is similar to a transpose-symmetric counterpart \cite{santra2002non,JoglekarSaxena}.
Given a pair of Hamiltonian $H$ and metric $\eta$, an intertwiner for any similar Hamiltonian, $S^{-1} H S$, can be constructed as~\cite{kretschmer2001interpretation}, 
\begin{equation}
(H,\eta)\leftrightarrow (H',\eta')=(S^{-1}HS,\eta'=S^\dag \eta S). \label{MetricMapper}
\end{equation}
Thus, the intertwiner of \eqref{homomorphismMetric} is readily generalized to Hamiltonians that are not transpose-symmetric, including cases with non-reciprocal couplings \cite{Hatano1996}.

\subsection{Equivalent Hermitian Hamiltonian} \label{Equivalent Hamiltonian Section}

Numerical eigensolvers that exploit a matrix’s structural properties may have improved efficiency and numerical robustness. In particular, the Hermitian eigenvalue problem admits specialized algorithms that are generally simpler and less computationally demanding than algorithms designed to treat general non-Hermitian matrices \cite{golub2013matrix}. Since similarity transforms preserve the spectra of matrices, computing the spectrum of a quasi-Hermitian matrix may be simplified if a similarity transformation that maps it to a Hermitian matrix is known. Such a similarity transformation can be found by taking a square root of a metric operator, $\eta$; choosing any matrix $S$ such that $S S^\dag = \eta^{-1}$ in \cref{MetricMapper} maps a quasi-Hermitian matrix to a similar Hermitian matrix \cite{williams1969operators,mosta2003equivalence}.

In the $\mathcal{PT}$-unbroken domain for $H$,  the metric operators of \cref{homomorphismMetric} with $|Z| < t_m$ define a family of equivalent Hermitian Hamiltonians,
\begin{align}
h:&= \Omega H \Omega^{-1} = h^\dag,
\end{align}
where the non-unitary similarity transformation, $\Omega=\eta^{1/2}$, is the unique positive square root of $\eta$,
\begin{align}
\Omega &= \begin{pmatrix}
\frac{\alpha}{2} \mathbbm{1}_m & \frac{Z^*}{\alpha t_m} \mathcal{P}_m \\
\frac{Z}{\alpha t_m} \mathcal{P}_m & \frac{\alpha}{2} \mathbb{1}_m
\end{pmatrix}=\frac{\alpha}{2}\mathbbm{1}_{2m}+\frac{1}{\alpha t_m}\left(\Re Z \sigma_x+\gamma \sigma_y\right)\otimes\mathcal{P}_m,
\end{align}
where $\alpha =\sqrt{1+|Z|/t_m}+\sqrt{1-|Z|/t_m}$ and $\otimes$ denotes the Kronecker-Zehfuss product \cite{Zehfuss1858}. When $Z = \gamma = 0$, the similarity transformation is trivial; otherwise, the explicit equivalent Hermitian Hamiltonian, $h$, is 
\begin{align}
h_{jk} &= t'_j \delta^k_{j+1} + {t'_j}^*\delta^j_{k+1} + \Re z_j \delta^j_k\\
t'_j &:= \begin{cases}
\sqrt{t_j t_{n-j}} & \, \text{if }j \neq m \\
\frac{\Re Z}{Z}t_m + \mathfrak{i} \frac{\gamma}{Z} \sqrt{t_m^2-|Z|^2} & \,\text{if }j = m.
\end{cases}\label{equivHermHam} 
\end{align}
The uniform, zero-detuning case with $\Re Z=0$ was conjectured in \cite[eq. (10)]{modak2021eigenstate}. Remarkably, $h$ remains tridiagonal and hence describes local interactions on a chain. This starkly contrasts with the generally nonlocal equivalent Hermitian Hamiltonians discussed in \cite{Korff2008}, whose interaction range diverges near an EP.


\subsection{\texorpdfstring{$\mathcal{C}$}{C} symmetry}

A transpose-symmetric Hamiltonian with unbroken $\mathcal{PT}$-symmetry admits a hidden $\mathcal C$-symmetry; this symmetry defines an inner product rendering $H$ self-adjoint and allows for $H$ to be interpreted as the Hamiltonian of a closed quantum system \cite{bender2002complex,CorrectCPT}. More generally, even if $H$ is not necessarily transpose-symmetric, a symmetry may be constructed from a pair of two of its intertwiners, $\eta_1$ and $\eta_2$, since $\eta_2^{-1} \eta_1$ commutes with $H$~\cite{BiOrthogonal}. In more general cases where $H$ may not be transpose-symmetric, a symmetry may be constructed if two intertwiners, $\eta_1$ and $\eta_2$, are known, since $\eta_2^{-1} \eta_1$ commutes with $H$~\cite{BiOrthogonal}. Taking $\eta_1 = G$ and $\eta_2 = \mathcal{P}$ gives
\begin{align}
\mathcal{C} := \frac{t_m}{\sqrt{t_m^2 - \gamma^2}} \mathcal{P} G = 
\frac{\mathfrak{i} \gamma \sigma_z \otimes \mathbb{1}_m + t_m \sigma_x \otimes\mathcal{P}_m}{\sqrt{t_m^2-\gamma^2}}.
\label{generalC}
\end{align} 
When $\gamma < t_m$, $\mathcal{C}$ commutes with $H$ and $\mathcal{PT}$ and is a non-Hermitian involution satisfying $\mathcal{C}^2 = \mathbb{1}$. Due to the $\mathcal{C}$ symmetry and non-degeneracy of $H$ \cite[Thm. 4.2]{elliott1953characteristic}, the eigenvectors of $H$ reside in the eigenspaces  of $\mathcal{C}$, 
\begin{align}
V_{\pm} = \text{span} \left\{t_m  e_k +(-\mathfrak{i}\gamma\pm \sqrt{t^2_m - \gamma^2} ) e_{\bar{k}} \,|\, k \in \{1, \dots, m \}\right\}.
\end{align}
The coalescence of $V_+$ and $V_-$ as one approaches $\gamma =t_m$ is a signature that $\gamma = t_m$ is an exceptional point.

\subsection{Exact spectra for detuned uniform chains} \label{sec:exact-spectra}
Below, we outline a process that determines a subset of the spectra in closed-form for uniform ($t_j=t>0$) chains with defect-localized detuning. 
We utilize the (re-scaled) characteristic polynomial determined in \cite{ortega2019mathcal},
\begin{align}
p_{H/t}(2x) 
&= U_n\left(x\right) - (\zeta_m + \zeta_{\bar{m}}) U_{n-m}\left(x\right)  U_{m-1}\left(x\right) \nonumber \\
&+ \zeta_m \zeta_{\bar{m}} U_{n-2m}\left(x\right) U_{m-1}\left(x\right)^2  \label{charPolySpecific},
\end{align} 
where $\zeta_m=z_m/t$ is a re-scaled defect strength and $U_{n \in \mathbb{Z}}(x)$ denotes a \textit{Chebyshev polynomial of second kind}. These orthogonal polynomials may be defined either trigonometrically almost everywhere,
\begin{align}
    U_n(x) = \sin\left[(n+1) \arccos x\right]/\sin (\arccos x), 
\end{align} 
or by means of a linear recurrence relation
\begin{align}
    U_n(x) &= 2x U_{n-1}(x) - U_{n-2}(x) \label{eqn:Chebyshev-Recurrence}
\end{align}
with the initial conditions
\begin{align}
    U_{-1}(x) = 0 &\quad& U_0(x) = 1.
\end{align}
Special cases of the characteristic polynomial were found earlier in \cite{rutherford1948xxv,losonczi1992eigenvalues,Joglekar2010}. 

First, we derive the special cases with closed form spectrum shown in \cref{table}  by factorizing the characteristic polynomial. The mechanism that enables this factorization is the identity 
\begin{align}
U_{2m}(x) &= U^2_{m}(x) - U^2_{m-1}(x).
\end{align}
Additionally, we utilize the Chebyshev polynomials of the \textit{third} and \textit{fourth} kinds, which are 
\begin{align}
    V_n(x) &:= U_n(x) - U_{n-1}(x), \\
    W_n(x) &:= U_n(x) + U_{n-1}(x).
\end{align}
Then, the characteristic polynomial may be factorized in the following cases,
\begin{align}
    p_{H/t}(2x) = \begin{cases}
        U_m^2(x) & \text{if } z_m = z_{\overline{m}}^* = \pm it \\
      U_m(x) V_m(x) & \text{if } z_m = z_{\overline{m}}^*= e^{\pm i \pi/3}t  \\
      U_m(x) W_m(x) & \text{if } z_m = z_{\overline{m}}^* = e^{\pm 2 i \pi/3}t  \\
      V_m(x)^2 & \text{if } z_m = z_{\overline{m}}^* = (1 \pm i)t \\
      W_m(x)^2 & \text{if } z_m = z_{\overline{m}}^* = (-1 \pm i)t 
    \end{cases}. \label{eqn:factorization}
\end{align}
The cases of \cref{eqn:factorization} with $|\Im z_m|=t$ also follow from \cref{eqn:poly-sum-of-squares}. Since the roots of the Chebyshev polynomials admit closed-form expressions,
\begin{align}
    U_n^{-1}(\{0\}) &= \left\{\cos\left(\frac{j \pi}{n+1} \right)\,\,\,\,\,\,\,\,\,\,\,\,|\, j \in \{1, \dots, n\} \right\}, \label{eqn:Chebyshev-2nd-Roots} \\
    V_n^{-1}(\{0\}) &= \left\{\cos\left(\frac{(2j-1)\pi}{2n+1}\right)\,|\,j \in \{1, \dots, n\} \right\}, \label{eqn:Chebyshev-3rd-Roots} \\
    W_n^{-1}(\{0\}) &= \left\{\cos\left(\frac{2j \pi}{2n+1}\right)\,\,\,\,\,\,\,\,\,|\,j \in \{1, \dots, n\} \right\}, \label{eqn:Chebyshev-4th-Roots}
\end{align}
the factorizations in \cref{eqn:factorization} may be used to compute the spectra of $H$; the result is displayed in \cref{table}. 

\begin{table*}[htp!]
\centering
\begingroup
\setlength{\tabcolsep}{10pt}
\renewcommand{\arraystretch}{1.75}
\begin{tabular}{|l|l|}
\hline
Constraints & Eigenvalues of $H$ ($j \in \{1, \dots, m\}$) \\
\hhline{|=|=|}
$z_m  = \pm \mathfrak{i} t$
& $\Bigl\{
2t \cos \Bigl(\dfrac{j \pi}{m+1}\Bigr)
\Bigr\}$\\[5pt]
\hline
$z_m = (-1\pm \mathfrak{i}) t$  
& $\Bigl\{
2 t\cos\Bigl( \dfrac{2j \pi }{2m + 1} \Bigr)\Bigr\} $\\[5pt]
\hline
$
z_m = (1\pm \mathfrak{i}) t$  
& $\Bigl\{
2t \cos \Bigl(\dfrac{(2j-1)\pi}{2m+1} \Bigr) \Bigr\}$ \\[5pt]
\hline
$z_m = e^{\pm \mathfrak{i} \pi/3} t$  
& $\Bigl\{ 2t \cos \Bigl(\dfrac{j \pi}{m+1} \Bigr)\Bigr\} \cup \Bigl\{2t \cos \Bigl(\dfrac{(2j-1) \pi }{2m + 1}\Bigr)\Bigr\} $\\[5pt]
\hline
$
z_m = e^{\pm 2 \pi \mathfrak{i}/3} t $ 
& $\Bigl\{2t \cos \Bigl(\dfrac{j \pi}{m+1}\Bigr) \Bigr\} \cup \Bigl\{2t \cos \Bigl(\dfrac{2j \pi}{2m+1}\Bigr)\Bigr\}$ \\[5pt]
\hline
\end{tabular}
\centering
\caption{Non-Hermitian cases where the spectrum of $H$ has a closed form solution. The lattice is even, uniform, and open with nearest-neighbour defects ($n = 2m$), and defect-localized detuning. The entry in the first row was known to \cite[p.3]{Babbey}. The eigenvectors can be constructed using the ansatz of \cref{eq:eigenvector-ansatz}.}
\label{table}
\endgroup
\end{table*}

If $\lambda$ is a known eigenvalue of $H$, a corresponding eigenvector
may be constructed by treating the eigenvalue problem as a three-term recurrence relation \cite{Gantmacher2002}. 
This ansatz is a simplified version of those given in \cite{Babbey,ortega2019mathcal}, with explicit prefactors and a simplified form beyond the second defect obtained using the characteristic equation.
\begin{equation}
\frac{\psi_j(\lambda)}{\psi_1(\lambda)}
=
\begin{cases}
U_{j-1}(x),
    & 1\leq j\leq m,\\[4pt]
U_{j-1}(x)
-\zeta_m U_{m-1}(x)U_{j-m-1}(x),
    & m \leq j\leq \bar m,\\[4pt]
D_{n-1}(x)\,U_{n-j}(x),
    & \bar m<j\leq n,
\end{cases}
\label{eq:eigenvector-ansatz}
\end{equation}
where $2 t x = \lambda$, as in \cref{charPolySpecific}, and
\begin{align}
D_{n-1}(x)
={}&U_{n-1}(x)
-\zeta_m U_{m-1}(x)U_{n-m-1}(x)
\nonumber\\
&-\zeta_{\bar m}U_{m-2}(x)
\left[
U_{\bar m-1}(x)
-\zeta_m U_{m-1}(x)U_{\bar m-m-1}(x)
\right].
\label{eq:D-n-minus-one}
\end{align}
If $H$ is $\mathcal{PT}$-symmetric and $\lambda=2tx$ is real, then a
corresponding eigenvector may be chosen to be an eigenvector of
$\mathcal{PT}$. Hence $|\psi_n|=|\psi_1|$, and therefore $|D_{n-1}(x)|=1$.
This conclusion need not hold for complex eigenvalues in the
$\mathcal{PT}$-broken phase.


A subset of the spectrum can be written in closed form for arbitrary defect positions as well.
If $\lambda/2t$ is a common root of polynomials $U_n, U_{n-m}, $ and $U_{n-2m}$, then this root is a real eigenvalue of $H$ that does not depend on the complex defect potential, $z_m, z_{\bar{m}}$.  Similarly, if $z_m = - z_{\bar{m}}$ and $\lambda/2t$ is a common zero of $U_n$ and $U_{n-2m}$, then $\lambda$ is a real eigenvalue independent of $z_m$. Thus, a subset of the spectrum is 
\begin{align}
d := \begin{cases}
\gcd(2 m,n+1) & \text{ if } z_m = -z_{\bar{m}}\\
\gcd(m,n+1) & \text{ otherwise},
\end{cases} \\
\sigma(H) \supseteq \left\{ 2t \cos \left(\pi \frac{k}{d}\right): k \in \{1, \dots d-1\} \right\}. 
\end{align} 
Notably, the above expression implies the existence of a zero-energy state whenever $n$ is odd and $z_m=-z_{\overline{m}}$. The existence of this zero mode may also be deduced using particle-hole and $\mathcal{PT}$ symmetries \cite{Joglekar2010}. When $m$ is even, the zero mode's eigenvector vanishes on both defect sites.

\section{Non-Hermitian SSH model with farthest defects} \label{sec:SSH}

In this section, we consider the SSH model perturbed by defect potentials. This model's tunnelling strengths are 2-periodic, meaning $t_k=t_{k'}$ when $k'\equiv k\mod 2$. Furthermore, we position the defects on the edges of the chain, so that $m=1$ and $\bar{m}=n$, and assume detuning is defect-localized. Our objective is to characterize the exceptional points and the number of real eigenvalues.

To simplify subsequent analysis, we rescale variables and "set $t_2 = 1$" as follows:
\begin{align}
    r := \frac{t_1}{t_2} &\quad& r_{L, R} := \frac{t_{L, R}}{t_2}&\quad& \zeta_j := \frac{z_j}{t_2} &\quad& g := \frac{\gamma}{t_2}  &\quad& \mu := \frac{\Delta}{t_2}  &\quad& \hat{H} = \frac{H}{t_2}. \label{eqn:rescale-variables}
\end{align}


$\mathcal{PT}$-symmetry is explicitly broken when $n$ is odd, $r \neq 1$, and $g \neq 0$. This leads to a complex spectrum for reasons explored in \ref{app:odd}. Thus, most of the writing here only considers even or critical chains.

\subsection{Bulk bands of the Hermitian SSH chain}

We start by recalling the bulk spectral properties of the SSH model $(\zeta_j = 0)$, as computed in \cite{SSH}. When $n$ is even and the boundary conditions are periodic $(r_{L, R} = 1)$, the resulting translation symmetry allows the Hamiltonian to be mapped to a Bloch Hamiltonian by means of a discrete Fourier transform,
\begin{align}
    \hat{H}(k)=
    \begin{pmatrix}
        0 & r + e^{-ik}\\
        r+e^{ik} & 0
    \end{pmatrix}, &\quad& k \in 4 \pi/n \{0, \dots, n/2-1\}
\end{align}
with dispersion relation
\begin{equation}
    \epsilon_\pm(k)
    =
    \pm\sqrt{1+r^2+2r\cos k}.
\end{equation}
Consequently, in the thermodynamic limit $n \to \infty$, the Hamiltonian's spectrum is the union of two sets called the \textit{valence} $(-)$ and \textit{conduction} $(+)$ bands,
\begin{equation}
    B_\pm
    :=
    \left\{
        \lambda\in\mathbb{R}:
        |1-r|\leq\pm\lambda\leq 1+r
    \right\}.
\end{equation}
There is a nonzero band-gap away from criticality $(r \neq 1)$. Thus, if the flow of electrons in a material is modeled by a non-critical SSH Hamiltonian, this material is an insulator. 

Perturbing the SSH model with endpoint potentials or alternative boundary conditions does not alter these bulk bands in the thermodynamic limit. However, these perturbations may introduce isolated eigenvalues residing outside the bands. For reasons that will be elaborated on in \cref{sec:edges}, these eigenvalues correspond to edge states. 





\subsection{Exceptional Contour Plots}

\Cref{edgeStateCounter} shows the $\mathcal{PT}$-phase portrait in the thermodynamic limit and sections of the exceptional contour for finite chains are plotted in \cref{fig:PTBreaking,thirdOrderFig}. The following subsections address the mathematical validity of these figures and their captions in detail; \Cref{sec:evalue-inclusion,sec:edges} address the numbers of real and imaginary eigenvalues, respectively. 

Henceforth, unless otherwise stated, we assume 
\begin{align} 
\zeta_1 &= \zeta_n^*, \label{eqn:PT-potential}\\
r_L&=-r_R,\qquad r_L^2\in\mathbb R. \label{eqn:PBC}
\end{align}
The constraints on $r_L, r_R$ allow for an open chain ($r_L=0=r_R$), a closed chain with purely imaginary Hermitian couplings ($i r_L=- i r_R \in \mathbb{R}$), or a closed non-Hermitian chain ($r_L=-r_R\in\mathbb{R}$).
Under these assumptions, the characteristic polynomial of \cref{charPolyTable} depends on the boundary couplings and gain–loss strength only through the effective parameter
\begin{align}
\widetilde g^{2}:=g^2+r_L^2.
\end{align}
The constraint \cref{eqn:PT-potential} implies the diagonal part of the Hamiltonian is $\mathcal{PT}$-symmetric. We temporarily relax these boundary conditions in \cref{esystemSection} to derive the eigensystem for arbitrary corner elements. 

\begin{figure}[!htbp]
\centering
\includegraphics[width = .95\textwidth]{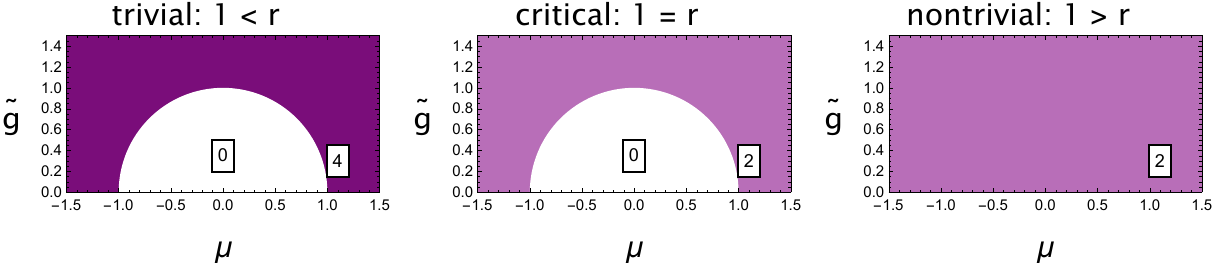}
\caption{
$\mathcal{PT}$-phases in the thermodynamic limit of even non-Hermitian SSH chains. The panels show the trivial $(r>1)$, critical $(r=1)$, and nontrivial $(r<1)$ bulk phases. White denotes an entirely real spectrum with no edge states; purple regions are labeled by the number of non-real eigenvalues, each corresponding to an edge state when the chain is open. \Cref{tab:edge-regions} defines these regions through inequalities. The critical phase diagram also applies to odd chains \cite{Korff2008,Willms2008}.} \label{edgeStateCounter}
\end{figure}

\Cref{fig:PTBreaking} shows the exceptional contours at $r=1$, computed as a parametric plot using proposition~\ref{prop:EP2-parameter}. Finite-size effects introduce complexity in the phase diagram that is absent in the thermodynamic limit.


\begin{figure}[!htbp]
\centering
\includegraphics[width=\textwidth]{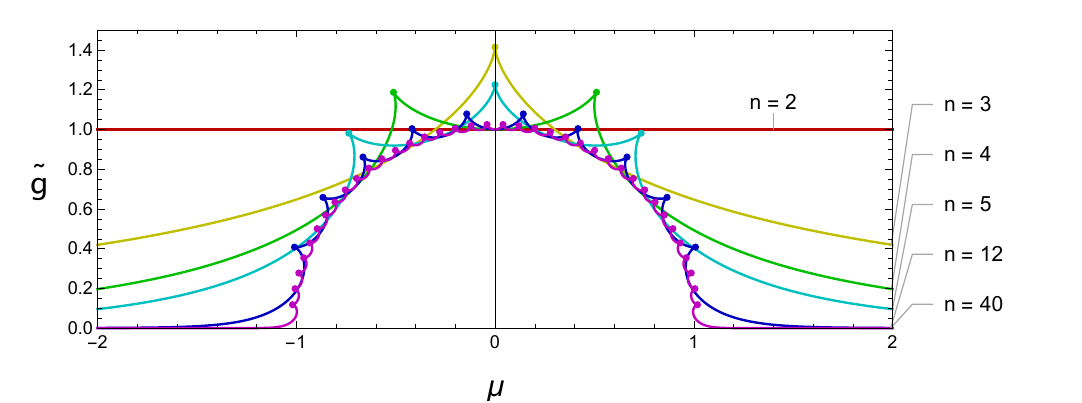}
\caption{Exceptional points (EPs) of an $n$-site uniform chain with boundary defect potentials. A pair of eigenvalues exhibits a real-to-complex spectral transition as the parameters are tuned above the EP contour. The contours are algebraic curves with cusp singularities marked by filled circles; we conjecture 
that for any $n$, the number of cusps with $\tilde{g} > 0$ is $n-2$. The cusp singularities are third-order EPs while the rest of the curve consists of second-order EPs.  The curves for $n=3,4,5$ are shown in \cite{Ruzicka2015}.}
\label{fig:PTBreaking}
\end{figure}

The numbers of non-real eigenvalues partition  parameter space into a disjoint union of three subsets, $\mathcal{R}_0, \mathcal{R}_2,$ and $\mathcal{R}_4$, as depicted in \cref{edgeStateCounter}. We explicitly define these subsets in \cref{tab:edge-regions}.
\begin{table}[!htbp]
    \centering
    \begin{tabular}{clcc}
        \toprule
        \shortstack{Region \\ \phantom{($r_L=0$, $n\to\infty$)}} & \shortstack{Parameter constraints \\ \phantom{($r_L=0$, $n\to\infty$)}}
        & \shortstack{Edge states\\($r_L=0$, $n\to\infty$)}
        & \shortstack{Real eigenvalues\\(finite $n$)} \\
        \midrule
        $\mathcal R_0$
            & $r\geq 1,\ \,\,\, \mu^2+\tilde g^2<1$  
            & $0$ & $n$ \\
            & $r^2 \geq \tilde{g}^2, \ \mu^2+\tilde g^2=1$
            & & \\[3pt]
        \midrule
        $\mathcal R_2$
            & $r=1,\ \mu^2+\tilde g^2>1,$
            & $2$ & $\geq n-2$ \\
            & $r<1,\ \mu^2+\tilde g^2\neq 1$ or $\tilde{g}^2 > r^2$
            & & \\[3pt]
        \midrule
        $\mathcal R_4$
            & $r>1,\ \mu^2+\tilde g^2>1$
            & $4$ & $\geq n-4$ \\
        \bottomrule
    \end{tabular}
    \caption{Regimes of parameter space (with $\tilde{g}^2 > 0$) classified by the edge state count for open chains $(n \to \infty)$ and a lower bound on the number of real eigenvalues (counted with algebraic multiplicity) for finite even chains.}
    \label{tab:edge-regions}
\end{table}




The full exceptional contour is a two-dimensional algebraic surface. We numerically sample points from this surface by computing roots of the \textit{discriminant} of the characteristic polynomial, where the discriminant of a polynomial is a well-known function of its coefficients that vanishes whenever the polynomial has a multiple root; the result is displayed in \cref{thirdOrderFig}. 

\begin{figure}[htp!]
   \centering

\subfigure[]{
\fbox{\includegraphics[width = .45 \textwidth,height = .2 \paperheight]{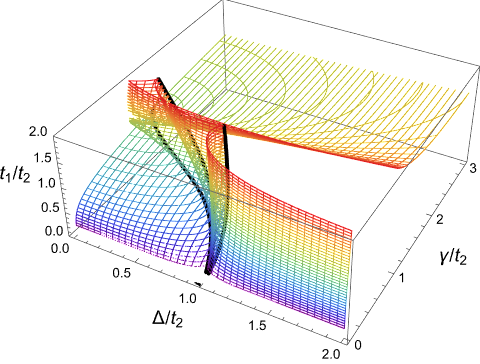}}
}
\subfigure[]{
\fbox{\includegraphics[width = .45 \textwidth,height = .2 \paperheight]
{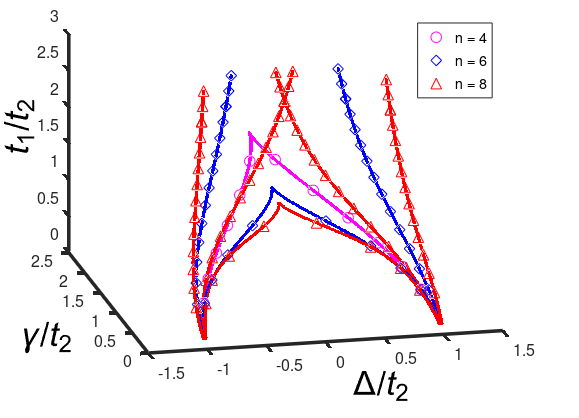}}
}
   
   \caption{(a) EPs of an open SSH chain with end-defects and $n = 8$. The domain satisfying the inequalities of \cref{ineq-1,ineq-2} lies in between this EP2 surface and the $\gamma = 0$ plane. Ridges of this surface correspond to EP3s, plotted in \cref{thirdOrderFig}(b). 
(b) Contours of EP3s residing in the exceptional surface. These contours exhibit cusp singularities at $\Delta = 0$ that are EP4s.}
   \label{thirdOrderFig}
\end{figure}

\subsection{Eigensystem solution} \label{esystemSection}
We start by computing the characteristic polynomial. The characteristic polynomial of the tridiagonal matrix corresponding to the Hermitian SSH model with open boundary conditions may be computed by solving a three-term linear recurrence relation \cite{LovassNagy1963,LovassNagy1964,Elsner1967},
\begin{align}
r^{-k} \text{det} \begin{pmatrix}
\epsilon & -r & 0 & \dots  \\
-r & \epsilon & -1 & \ddots  \\
0 & -1 & \epsilon & -r  \\
\vdots & \ddots & -r & \ddots \\
\end{pmatrix} 
&= \begin{cases}
\epsilon U_k(Q) & \text{ if } n = 2k + 1 \\
U_k(Q) + r^{-1} U_{k-1}(Q) & \text{ if } n = 2k.
\end{cases}
\end{align}
where we define 
\begin{align}
k &:= \floor{n/2}\\
Q &:= \frac{\epsilon^2 - r^2 - 1}{2r}. \label{eqn:Q}
\end{align}

The multilinearity property of determinants for rows and columns allows this result to be generalized to non-Hermitian cases ($\zeta \neq 0$) with corner elements ($r_{L,R}\neq 0$). The resulting polynomials are 
\begin{align}
    r^{-k}p_{\hat{H}}(\epsilon) = \begin{cases}
       \begin{array}{l}
U_{k}(Q) + \left(\zeta_1 \zeta_n - r_L r_R\right)U_{k-2}(Q) - r_L - r_R\\
+ r^{-1} \left(1 - \epsilon(\zeta_1 + \zeta_n) + \zeta_1 \zeta_n - r_L r_R \right) U_{k-1}(Q)  
\end{array}  & \text{if } n = 2k\\[16pt]
      \begin{array}{l}
\left(\epsilon - \zeta_1 - \zeta_n\right) U_k(Q) - r_L- r_R \\
+  r^{-1} \left(\epsilon(\zeta_1 \zeta_n - r_L r_R)- \zeta_1 r^2 - \zeta_n \right)U_{k-1}(Q)
\end{array}  & \text{if } n = 2k+1
    \end{cases} \label{charPolyTable}
\end{align}
The case for an open chain was known to \cite{da2007characteristic} and the uniform case, $r=1$, was known to \cite{YUEH2008}.


The eigenvalues of $\hat{H}$ may be found by locating the roots of the characteristic polynomial. Once an eigenvalue is known, a corresponding eigenvector, $\psi$, may be constructed from this eigenvalue using the following lemma. 

\begin{lemma} \label{lemma:eigenstates}
    The complex vector $\psi \in \mathbb{C}^n$ satisfies $\psi \in \ker(\epsilon \mathbb{1}_n- \hat{H})$ if and only if 
\begin{align}
    M(\epsilon) \begin{pmatrix}
        \psi_2 \\
        \psi_1
    \end{pmatrix} = 0, \label{eqn:relate-psi-1-and-2}
\end{align}
where $M(\epsilon)$ is the $2 \times 2$ matrix of \cref{eqn:transfer-2}, and the remaining elements in $\psi$ are related to the first two elements via
\begin{align}
\psi_{2j} &= \left(U_{j-1}(Q) + r U_{j-2}(Q) \right) \psi_2 -\epsilon U_{j-2}(Q) \psi_1 &\quad& \forall j \in \{1, \dots, \floor{n/2} \} \label{eqn:eigenstate-even}\\
\psi_{2j-1} &= \epsilon U_{j-2}(Q) \psi_2 - \left(r U_{j-2}(Q) + U_{j-3}(Q) \right) \psi_1 &\quad& \forall j \in \{1, \dots, \ceil{n/2} \} \label{eqn:eigenstate-odd}
\end{align}
Since $p_{\hat{H}}(\epsilon) = - r^{k-1} \det M(\epsilon)$, $\epsilon$ is an eigenvalue of $\hat{H}$ if and only if there exists a nonzero $\psi$ satisfying the preceding two equations
\end{lemma}
\begin{proof}
The eigenvalue problem may be considered as a linear recurrence relation,
\begin{align}
t_j\psi_{j+1} = t_2 \epsilon \psi_j - t_{j-1} \psi_{j-1}, &\quad& \forall j \in \{2, \dots n-1\}. \label{eqn:SSH-bulk-evalue-eqn}
\end{align}
The linear recurrence relation in $\psi$ may be solved using arithmetic with $2 \times 2$ matrices. Most of the elements of $\psi$ may be determined from the first two by computing a $2 \times 2$ 
matrix power, since
\begin{align}
\begin{pmatrix}
\psi_{2j} \\
\psi_{2j - 1}
\end{pmatrix} &= \begin{pmatrix}
\frac{\epsilon^2-1}{r} & -\epsilon \\
\epsilon & - r
\end{pmatrix}^{j-1} \begin{pmatrix}
\psi_2 \\
\psi_1
\end{pmatrix} &\quad& \forall j \in \{1, \dots, k \}. \label{eqn:step-via-power}
\end{align}
Owing to the Cayley-Hamilton theorem, the powers of an invertible $2 \times 2$ matrix, $A$, may be computed using Chebyshev polynomials; the result is \cite{Ricci1975}
\begin{align}
A^j &=(\det A)^{j/2}\left[-U_{j-2}(y)\mathbbm{1}_2+ U_{j-1}(y)\frac{A}{\sqrt{\det A}}\right] &\quad& \forall j \in \mathbb{N}, \label{eqn:2x2-power}
\end{align}
where $y=\text{tr} A/(2\sqrt{\det A})$ is a dimensionless argument. Then, computing the matrix power in \cref{eqn:step-via-power} via the preceding \cref{eqn:2x2-power} determines the eigenstate elements, as in \cref{eqn:eigenstate-even,eqn:eigenstate-odd}, everywhere except at the site $n = 2k+1$ when $n>1$ is odd. The determination of this last element follows from inserting the expressions for $\psi_{2k}, \psi_{2k-1}$ into the eigenvalue equation, \cref{eqn:SSH-bulk-evalue-eqn} at site $j = n-1$.

The eigenvalue equations at the end sites, relating the elements $\psi_1, \psi_2, \psi_{n-1}, $ and $\psi_n$, may be compactly written as \cref{eqn:relate-psi-1-and-2} with
\begin{align}
    M(\epsilon) = \begin{pmatrix}
        r_L & 0 \\
        \mu - i g -\epsilon & \frac{t_{n-1}}{t_2}
    \end{pmatrix} \begin{pmatrix}
        \epsilon & -r \\
        1 & 0
    \end{pmatrix}^{n-2k} \begin{pmatrix}
\frac{\epsilon^2-1}{r} & -\epsilon \\
\epsilon & - r
\end{pmatrix}^{k-1} + \begin{pmatrix}
    r & \mu + i g - \epsilon \\
    0 & r_R
\end{pmatrix} \label{eqn:transfer-1}
\end{align}
after eliminating the variables $\psi_{n-1}, \psi_n$ with \cref{eqn:eigenstate-even,eqn:eigenstate-odd}. Equivalently, using \cref{eqn:2x2-power} to compute the matrix power in the preceding equation, we may write
\begin{align}
M(\epsilon) =
\begin{cases}
\begin{pmatrix}
r+r_L(u+rv) & \zeta_1 - \epsilon-r_L\epsilon v \\
(\zeta_n - \epsilon)(u+rv)+r\epsilon v &
r_R-(\zeta_n-\epsilon)\epsilon v-r(rv+w)
\end{pmatrix}
& \text{if } n=2k, \\[4mm]
\begin{pmatrix}
r+r_L\epsilon u & \zeta_1-\epsilon-r_L(ru+v) \\
(\zeta_n-\epsilon)\epsilon u+u+rv &
r_R-(\zeta_n-\epsilon)(ru+v)-\epsilon v
\end{pmatrix}
& \text{if } n=2k+1,
\end{cases}, \label{eqn:transfer-2}
\end{align}
where $u, v, w = U_{k-1}(Q), U_{k-2}(Q), U_{k-3}(Q)$, respectively.


\end{proof}
An eigenvalue has a geometric multiplicity of two if and only if \cref{eqn:relate-psi-1-and-2} holds trivially due to $M(\epsilon) = 0$. If $M(\epsilon) \neq 0$, then \cref{eqn:relate-psi-1-and-2} determines the unique eigenvector corresponding to $\epsilon$ up to its normalization. When $r_L \neq r_R$ or $r_L = r_R = 0$, an eigenvalue always has a geometric multiplicity of one due to proposition~\ref{prop:geometric-multiplicity}.

\subsection{Real eigenvalue inclusion results} \label{sec:evalue-inclusion}


This section is devoted to finding subsets of the complex plane that contain some or all of the eigenvalues of $H$. By examining the real axis, we find parameter regimes for which $H$ has a real spectrum. We start by graphically summarizing our search in \cref{fig:Intervals}. 
\begin{figure}[htp!]
    \centering
    \includegraphics[width=0.9\linewidth]{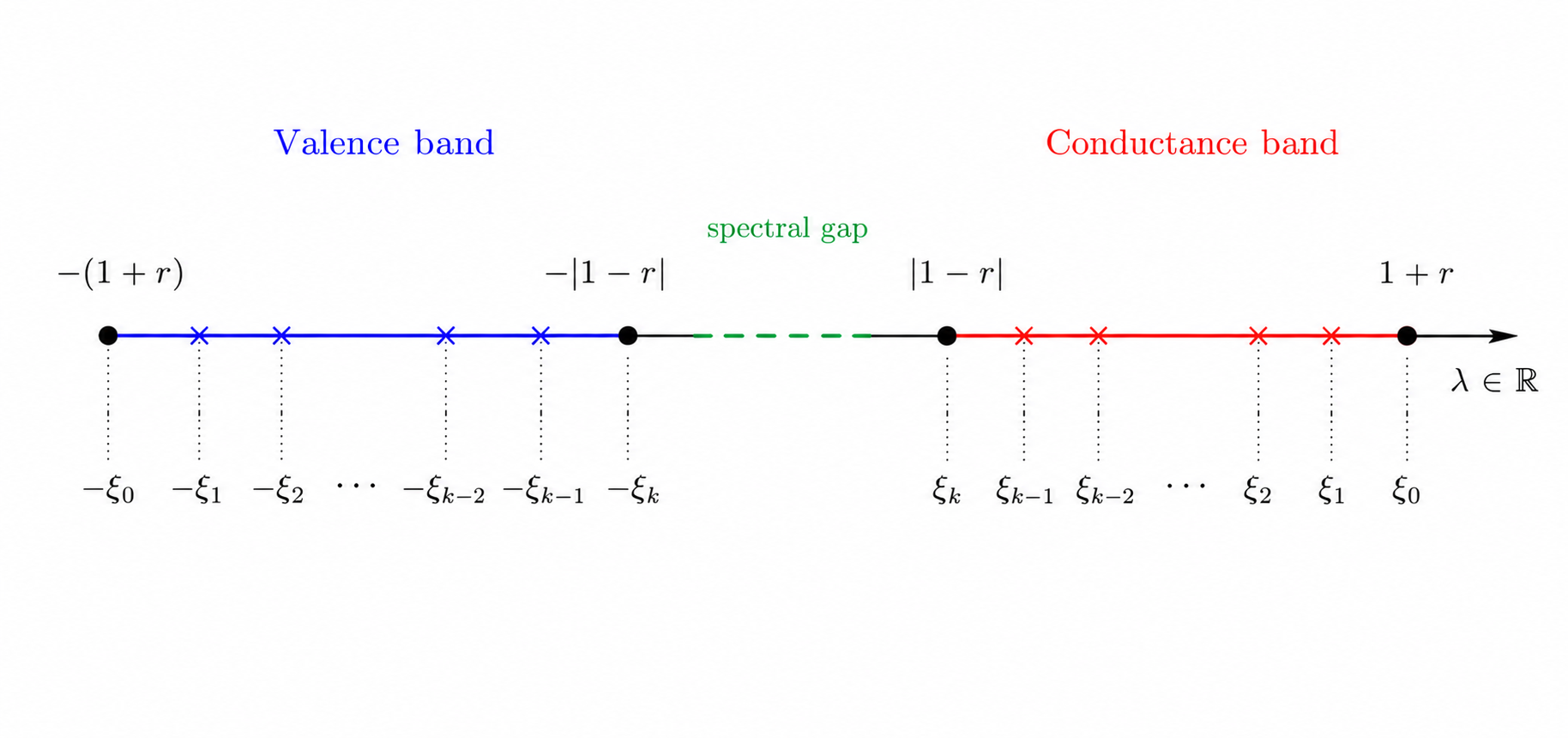}
    \caption{
    Schematic of strategy to locate eigenvalues. We subdivide the conduction and valence bands of the Hermitian SSH model with open boundary conditions into intervals $\mathcal{I}_\pm(j)$ with endpoints $\xi_j$ defined in \cref{eqn:intervals,eqn:interval-endpoints}. \Cref{inclusionTheorem} presents conditions that guarantee the existence of real eigenvalues in these intervals or in the band gap.}
    \label{fig:Intervals}
\end{figure}

Our choice for a subdivision of search intervals in the bulk bands is motivated by fundamental properties of the Chebyshev polynomials. The roots of $U_{n-1}$ and $U_n$ are \textit{strictly-interlacing}, a fact that may be seen by direct inspection of \cref{eqn:Chebyshev-2nd-Roots} and readily generalizes to other families of orthogonal polynomials \cite[Prop. 6.3.10]{Rahman2002}; 
we remind the reader that if \(p\) and \(q\) are real-rooted polynomials such that $\deg p = \deg q + 1$ and their zeros are respectively sorted in ascending order,
\[
x_1 < x_2 < \cdots < x_m
\qquad\text{and}\qquad
y_1 < y_2 < \cdots < y_{m-1},
\]
then the zeros of \(q\) \emph{strictly interlace} those of \(p\) if
\[
x_1 < y_1 < x_2 < y_2 < \cdots < y_{m-1} < x_m.
\]
In other words, the roots of $U_{n-1}$ are contained in intervals whose boundary is comprised of roots of $U_n$. This motivates an effort to find criteria that locate roots of the characteristic polynomial in intervals bounded by roots of a suitable Chebyshev polynomial. A fruitful choice happens to be $U_{k-1} \circ Q$. 
To achieve this, we define the open intervals
\begin{align}
    \mathcal{I}_\pm(j) := \{x \in \mathbb{R} \,| \, \xi_{j+1} < \pm x < \xi_{j} \} &\quad& \forall j \in \{0, \dots, k-1\}, \label{eqn:intervals}
\end{align}
where 
\begin{align}
    \xi_j := |r + e^{2 \pi \mathfrak{i} j/n}| &\quad& \forall j \in \{0, \dots, k\}. \label{eqn:interval-endpoints}
\end{align}
Then, when $n$ is even, the roots of $U_{k-1} \circ Q$ are the endpoints of $\mathcal{I}_{\pm}(j)$ $\forall j \in \{1, \dots, k-2\}$, i.e.
\begin{align}
    (U_{k-1}\circ Q)^{-1}(\{0\}) = \{\pm \xi_j \,|\, j \in \{1, \dots, k-1\} \}. \label{eqn:roots-of-UQ}
\end{align}

For open critical chains, some criteria that guarantee the existence of real eigenvalues in these intervals were determined by applying the intermediate value theorem for roots (Bolzano's theorem) in \cite[\S 3.2.4]{Willms2008}. Below, we discuss the extent to which this argument extends to non-critical or non-open chains.

\begin{theorem} \label{inclusionTheorem}
    Assume $n$ is even or $n$ is odd and the chain is critical $(r = 1)$. Let $\beta = \mu^2 + \tilde{g}^2$. Then,
        \begin{enumerate}
        \item when $\beta = 1$, the spectrum of $H$ admits the following closed-form expression,
        \begin{align}
            \sigma(\hat{H}) = \{\pm \xi_j \,|\, j \in \{1, \dots, \floor{\frac{n-1}{2}}\} \} \cup \begin{cases}
                 \left\{\mu \pm \sqrt{r^2 - \tilde{g}^2} \right\} &\text{if } n = 2k,\\ \{2 \mu\} &\text{if } n = 2k+1, r = 1.
            \end{cases} \label{exactSpectrum}
        \end{align}
        \item when $\beta \neq 1$, the interval $\mathcal{I}_\pm(j)$ contains a real eigenvalue of $H$ with odd algebraic multiplicity for all $j \in \{1, \dots, k-2\}$. When $n$ is odd and $\beta \neq 1$, a real eigenvalue with odd algebraic multiplicity exists in $(-\xi_k, \xi_k)$. \label{item:locate-band-evalues}
        \item when $n \geq 4$, real eigenvalues of $H$ with odd algebraic multiplicity exist in the intervals $\mathcal{I}_\pm(0)$ and $\mathcal{I}_\pm(k-1)$ when certain parametric constraints hold in addition to $\beta \neq 1$. We express these constraints in the following implications: 
        \begin{align}
            1 + \frac{n}{2} \left(1 + r^{-1} \right)\dfrac{1 \mp 2 \mu + \beta}{1 - \beta} > 0 \, &\Rightarrow \, \sigma(\hat{H}) \cap \mathcal{I}_\pm(0) \neq \emptyset \label{ineq-1}\\
            1 + k\left(1 - r^{-1} \right)\dfrac{1 \mp 2 \mu + \beta}{1 - \beta} > 0 \, &\Rightarrow \, \begin{cases} 
            \sigma(\hat{H}) \cap \mathcal{I}_\pm(k-1) \neq \emptyset &\text{if } r \leq 1 \\
            \sigma(\hat{H}) \cap \mathcal{I}_\mp(k-1) \neq \emptyset &\text{if } r \geq 1
            \end{cases} \label{ineq-2}
        \end{align}
        \Cref{ineq-1} also holds when $n = 3$. \label{thm:item-edges-of-bands}
        \item when $n$ is even, there is a zero mode if and only if $\beta =  r^{2k}$. This eigenvalue is simple if and only if $\mu \neq 0$. \label{thm:item-zero-mode}
        \item decompose the nonzero part of the spectral gap into two subintervals, 
\begin{align}
\mathcal{I}_\pm^{\mathrm{gap}}
:=\{\epsilon\in\mathbb R:
0<\pm\epsilon<\xi_k\}.
\end{align}
When $n$ is even, a nonzero real eigenvalue with odd algebraic multiplicity exists in $\mathcal{I}_\pm^{\mathrm{gap}}$ when certain parametric constraints hold, we express these as the following implication:
\begin{align}
\frac{1-\beta+k\left[(1-r^{-1})(1+\beta)
\pm 2r^{-1}\mu\xi_k\right]}{\beta-r^{2k}}>0
\quad\Longrightarrow\quad
\sigma(\widehat H)\cap \mathcal{I}_\pm^{\mathrm{gap}}
\neq\emptyset, \label{eqn:gap-eval-exists}
\end{align}
\item if $\tilde{g}^2 \leq 0$, then the eigenvalues are real and simple.
    \end{enumerate}
\end{theorem}
\begin{proof}
Bolzano's theorem states that if a continuous function's valuations at two distinct points are nonzero with opposite signs, then this function has a root of in the open interval bounded by these two points \cite{Bolzano1817,RUSS1980156}. Bolzno's theorem and the fundamental theorem of algebra imply that a polynomial has a root with odd algebraic multiplicity in these kinds of intervals. The assertions of this theorem are proven by observing these kinds of sign changes in the characteristic polynomial.

When $n = 2k$ is even, the characteristic polynomial may be rewritten as
\begin{align}
r^{-k} p_{\hat{H}}(\epsilon) &= \left(1- \beta  \right) U_k(Q)  
+ \left(2 \beta Q + \frac{1 - 2 \epsilon \mu + \beta}{r}  \right) U_{k-1}(Q) 
\label{altPoly}
\end{align}
by using the defining recurrence relation of \cref{eqn:Chebyshev-Recurrence}.
If the chain is critical $(r = 1)$, then the characteristic polynomial is
\begin{align}
    p_{\hat{H}}(\epsilon) = (\epsilon - 2 \mu) U_{n-1}\left(\frac{\epsilon}{2}\right) + 
         (\beta - 1) U_{n-2}\left(\frac{\epsilon}{2}\right). \label{altPoly-Odd}
\end{align}
When $n$ is even, the preceding equation may be derived from \cref{charPolyTable}
by using the composition identity of \cite{mason2002chebyshev,Zhang2004,Brandi2020}
\begin{align}
    U_{2j+1}(2^{-1} x) = x U_{j}\left(2^{-1} x^2 - 1 \right) &\quad& \forall x \in \mathbb{C}.
\end{align}
When the chain is open, \cref{altPoly-Odd} is equivalent to the $m = 1$ case of \cref{charPolySpecific} 


    \begin{enumerate}
        \item Suppose $\beta = 1$.
        When $n$ is even, \cref{exactSpectrum} follows from the observation that the first term of \cref{altPoly} vanishes and \cref{eqn:roots-of-UQ}. If, instead, $n$ is odd and the chain is critical, \cref{exactSpectrum} follows from the observation that the second term of \cref{altPoly-Odd} vanishes and \cref{eqn:roots-of-UQ}. 
        \item 
Next, consider the case $\beta \neq 1$. 
We will show that the sign of the characteristic polynomial is different at the endpoints of the intervals $\mathcal{I}_{\pm}(j)$ for all $j \in \{ 1, \dots, k-2 \}$. The relevant valuations of these polynomials may be deduced using, for example, the following elementary facts: \footnote{The first is logically equivalent to the \textit{equioscillation property} of the Chebyshev polynomial of the first kind.}
\begin{align}
    U_k(\cos(\pi j /k)) = (-1)^j &\quad& \forall j \in \{1, \dots, k-1\},
\end{align}
\begin{align}
    r = 1 \, \Rightarrow \, \xi_j = 2 \cos(\pi j/n) 
\end{align}
and the root sets of Chebyshev polynomials, as in \cref{eqn:roots-of-UQ,eqn:Chebyshev-2nd-Roots}.
Then, the aforementioned sign changes may be seen from the following equation,
\begin{align}
    r^{-k} p_{\hat{H}}(\pm \xi_j) &=  (\pm1)^n (- 1)^j \left(1- \beta \right) &\quad& \forall j \in \{1, \dots, \floor{\frac{n-1}{2}}\}. 
\end{align}
\item
The inequalities of \cref{ineq-1,ineq-2} follow from considering the sign of the characteristic polynomial at the endpoints of the intervals $\mathcal{I}_\pm(0)$ and $\mathcal{I}_{\pm}(k-1)$. 
A sign change may be observed by computing the ratios of the characteristic polynomial's valuation at the endpoints of $\mathcal{I}_\pm(0)$ and $\mathcal{I}_\pm(k-1)$, respectively,
\begin{align}
    -\frac{p_{\hat{H}}(\pm \xi_0)}{p_{\hat{H}}(\pm \xi_1)}
&= 1 + \frac{n}{2}\left(1 + r^{-1} \right)\dfrac{1 \mp 2 \mu + \beta}{1 - \beta}  \\
-\frac{p_{\hat{H}}(\pm \xi_k)}{p_{\hat{H}}(\pm \xi_{k-1})} 
&= 1 + k\dfrac{\left(1 - r^{-1} \right)(1+\beta) \pm 2 r^{-1} \mu |r-1|}{1 - \beta} 
\end{align}
These ratios need to be positive for the characteristic polynomial to change sign at the endpoints of the intervals.

\item This criterion for the existence of a zero mode directly follows from evaluating the characteristic polynomial at $\epsilon = 0$, 
\begin{align}
    (-1)^k \det(-\hat{H}) = r^{2k}-\beta.
\end{align}
The zero mode is simple if and only if it is not a root of the derivative of the characteristic polynomial. The derivative at $\epsilon = 0$ may be immediately deduced from \cref{altPoly},
\begin{align}
    \frac{\partial}{\partial \epsilon} p_{\hat{H}}(0) = -2 \mu r^{k-1} U_{k-1}\left(\frac{-r-r^{-1}}{2} \right).
\end{align}
We note this derivative equals zero if and only if $\mu = 0$ since 
\begin{align}
    r^{k-1} U_{k-1}\left(\frac{-r-r^{-1}}{2} \right) = \begin{cases} (-1)^{k-1} \frac{r^{2k-1} - 1}{r^2 - 1} & \text{if } r \neq 1 \\
    (-1)^{k-1} k & \text{if } r = 1.
    \end{cases}
\end{align}

\item This item's assertion follows using exactly the same logic used for \cref{thm:item-edges-of-bands}. In particular, note
\begin{equation}
-\frac{p_{\widehat H}(\pm\xi_k)}
{p_{\widehat H}(0)}
=
\frac{r^k}{\beta-r^{2k}}
\left[
1-\beta+k\left(
(1-r^{-1})(1+\beta)
\pm 2r^{-1}\mu\xi_k
\right)
\right].
\end{equation}
\item If $\tilde{g}^2 \leq 0$, then the characteristic polynomial coincides with that of an open Hermitian SSH chain with the defect potentials $\zeta_1 = \mu + \sqrt{-\tilde{g}^2}$ and $\zeta_n = \mu - \sqrt{-\tilde{g}^2}$, whose spectrum is real and simple owing to \cite[Thm. 4.2]{elliott1953characteristic}.
\end{enumerate}
\end{proof}

The inequalities of \cref{ineq-1,ineq-2} when $r = 1$ were known to  \cite[eq. (63-64)]{Willms2008} and special cases of \cref{exactSpectrum} were presented in \cite{rutherford1948xxv,Willms2008,Korff2008,guo2016solutions}. The zero mode of \cref{thm:item-zero-mode} in the uniform zero-detuning $(r = 1, \mu = 0)$ case was known to \cite{Korff2008,jin2009solutions}, where $g = 1$ is an EP2. 

The following corollary presents a simplified case of the above theorem and, when combined with \cref{exactSpectrum}, is sufficient to deduce the rightmost column of \cref{tab:edge-regions}.

\begin{corollary} \label{corollary}
    Assume $n$ is even or $n$ is odd and the chain is critical $(r = 1)$. Let $\beta = \mu^2 + \tilde{g}^2$ and assume $\tilde{g}^2 > 0$. 
    \begin{enumerate}
        \item There are at least $n-4$ distinct real eigenvalues of $\hat{H}$ with odd algebraic multiplicity, of which at least $2k-4$ reside in the bands $B_{\pm}$.  
        \item When $\beta < 1 \leq r$  (so the chain is trivial or critical), the spectrum is composed entirely of real simple eigenvalues residing in the bands $B_{\pm}$. 
        \item When $\beta \neq 1$ and the chain is critical or non-trivial $(r \leq 1)$, there are at least $n - 2$ distinct real eigenvalues of $\hat{H}$ with odd algebraic multiplicity. There is at most one non-simple eigenvalue. $\hat{H}$ is at an EP3 only if the non-simple eigenvalue resides in the bands. 
        \label{thm:item-one-nonsimple}
    \end{enumerate}
\end{corollary}
\begin{proof}
    \begin{itemize}
        \item If $\beta = 1$, this may be seen by directly inspecting \cref{exactSpectrum}. If $\beta \neq 1$, the aforementioned band eigenvalues are the eigenvalues located in the intervals $\mathcal{I}_\pm(j)$ in \cref{item:locate-band-evalues} of \cref{inclusionTheorem} and when $n$ is odd, an additional real eigenvalue exists in $(-\xi_k, \xi_k)$.
        \item This directly follows from the corollary's preceding item and the fact that \cref{ineq-1,ineq-2} hold.
        \item In addition to the aforementioned $n-4$ eigenvalues, when $r \leq 1$, at least two of the four intervals in $\{\mathcal{I}_\pm(0), \mathcal{I}_\pm(k-1)\}$ also have an eigenvalue with odd algebraic multiplicity. When $r = 1$, this is because \cref{ineq-2} holds automatically; otherwise, this is because $1+r^{-1}$ and $1-r^{-1}$ have opposite signs, so for each of the signs $\pm$, at least one of the inequalities \cref{ineq-1,ineq-2} automatically holds.
        Let $E$ denote the set of these $n-2$ distinct real eigenvalues of $\hat{H}$ with odd algebraic multiplicity. Suppose a non-simple eigenvalue, $\epsilon_*$, exists. If $\epsilon_* \not\in E$, since the sum of algebraic multiplicities of all eigenvalues is $n$, it follows that $\epsilon_*$ has algebraic multiplicity two and that $\epsilon_*$ is the only eigenvalue not in $E$. If $\epsilon_* \in E$, then its algebraic multiplicity is odd, and the same argument implies that its multiplicity is three and that $E$ is the complete set of eigenvalues of $\hat{H}$. $\hat{H}$ is at an EP3 if and only if $\epsilon_*$ has multiplicity three, or equivalently $\epsilon_* \in E$, and the rest of the corollary statement follows from the fact $E$ is a subset of the bands.
    \end{itemize}
\end{proof}

Tangentially, we remark on how the Brauer-Ostrowski ovals theorem may be used to derive a less sharp eigenvalue inclusion result in \ref{app:Cassini}. 


\subsection{Edge states of the even SSH chain with open boundary conditions} \label{sec:edges}
The eigenvectors of tight-binding models are commonly characterized as either \textit{bulk} or \textit{edge} states; 
roughly, the bulk eigenstates are extended throughout the chain irrespective of the chain size, whereas edge states remain exponentially localized near the boundaries. 

We first remark that an eigenstate is an edge state only if its eigenvalue does not reside in the bulk bands $B_\pm$, or equivalently, $Q \not\in [-1,1]$. To see why, we can determine the localization properties of an eigenstate through their relationship with the Chebyshev polynomials expressed in lemma~\ref{lemma:eigenstates}. Fix an eigenvalue, $\epsilon$, and define 
\begin{align}
    \rho_{\pm} := Q \pm \sqrt{Q^2 - 1}.
\end{align}
Since $\rho_+ = \rho_-^{-1}$ and
\begin{align}
U_j(Q) = \frac{\rho_+^{j+1} - \rho_-^{j+1}}{\rho_+ - \rho_-}, \label{eqn:chebyshev-poly-def}
\end{align}
the sequence $U_j(Q)$ is oscillatory in the site index $j$ when $|\rho_+| = 1$ and scales exponentially with $j$ when $|\rho_+| \neq 1$. Thus, the eigenstate corresponding to $\epsilon$ is a bulk state if and only if $|\rho_+| = 1$, which holds if and only if $Q \in [-1, 1]$ since $2Q$ is a Joukowski transform of $\rho_+$,
\begin{align}
    \rho_+ + \frac{1}{\rho_+} = 2 Q.
\end{align}


In this section, we seek to count the number of edge states in the even SSH chain with open boundary conditions; our result is summarized in \cref{tab:edge-regions}. Unlike the finite-size spectral results of the preceding sections, the analysis below is approximative in an asymptotic sense: we neglect terms that vanish exponentially with the chain length. We expect the resulting edge-state conditions to become exact in the thermodynamic limit $(n\to\infty)$.

The edge-states are characterized by a decay factor $\rho_*^{-1}$ with modulus less than one. The eigenvalue of these edge-states is determined and related to the decay factor in \cref{eqn:edge-evalue,eqn:lambda-to-rho} and is not real when $g \neq 0$. 


In the open boundary case, the relationship between the eigenstate amplitudes at the leftmost two sites, \cref{eqn:relate-psi-1-and-2}, simplifies to 
\begin{align}
\left[\zeta_1 - \lambda \right]\psi_1= -r \psi_2.
\end{align}
Inserting this identity into \cref{eqn:eigenstate-even,eqn:eigenstate-odd}, we find
\begin{align}
\begin{pmatrix}
\psi_{2k} \\
\psi_{2k - 1}
\end{pmatrix} 
&=  \begin{pmatrix}
\frac{\lambda-\zeta_1}{r} U_{k-1}(Q) - \zeta_1 U_{k-2}(Q) \\
U_{k-1}(Q) + \frac{1 - \lambda \zeta_1}{r} U_{k-2}(Q)
\end{pmatrix} \psi_1.
\end{align}

Suppose $|\rho_+| \neq 1$ and let 
\begin{align} 
\rho_* := \text{argmax}_{\rho \in \{\rho_+, \rho_-\}} |\rho|. \label{eqn:decay}
\end{align}
Then, owing to \cref{eqn:chebyshev-poly-def}, we can approximate $U_{k-1}(Q) = \rho_* U_{k-2}(Q) + O(\rho_*^{-k})$. Thus,
\begin{align}
\begin{pmatrix}
\psi_{2k} \\
\psi_{2k - 1}
\end{pmatrix}    &=  U_{k-2}(Q) \psi_1 \Psi + O(\rho_*^{-k}),
\end{align}
where 
\begin{align}
    \Psi = \begin{pmatrix}
\frac{\lambda-\zeta_1}{r} \rho_* - \zeta_1 \\
\rho_* + \frac{1 - \lambda \zeta_1}{r} 
\end{pmatrix}.
\end{align}
An eigenstate, $\psi$, is a left edge state if and only if its amplitudes decay exponentially in the site index as this index moves into the bulk. Thus, for $\psi$ to be an edge state, we must have $\Psi = 0$, which happens if and only if $\rho_*$ is a root of the quadratic equation 
\begin{align}
    r \rho_*^2 + \rho_*(1 - \zeta_1^2) = \zeta_1^2 r \label{eqn:edge-poly}
\end{align}
and
\begin{align}
    \lambda = \frac{1 + r \rho_*}{\zeta_1}.\label{eqn:lambda-to-rho}
\end{align}
Thus, 
\begin{align}
    \lambda = \frac{\zeta_1^2 + 1 \pm \sqrt{(\zeta_1^2 + 1)^2 + 4 \zeta_1^2(r^2 - 1)}}{2 \zeta_1}, \label{eqn:edge-evalue}
\end{align}
is an edge state's eigenvalue if and only if the corresponding $|\rho_*|$ given by \cref{eqn:lambda-to-rho} resides outside the unit disk. 

We proceed to deduce that an edge state's eigenvalue, $\lambda$, is never real when $g \neq 0$. Suppose $\lambda$ was real and $g \neq 0$. Then, since $\Psi = 0$ implies
\begin{align}
    \lambda^2 - (\zeta_1 + \zeta^{-1}_1) \lambda + 1 - r^2 = 0, \label{eqn:lambda-quadratic}
\end{align}
taking the imaginary part of both sides would imply $\lambda = 0$ or $\zeta_1^* = \zeta_1^{-1}$. By performing elementary manipulations on \cref{eqn:lambda-to-rho,eqn:lambda-quadratic}, both of these possibilities would imply $|\rho_*| = 1$ when $\zeta_1 \neq 0$, contradicting the assumption $|\rho_*| > 1$.

The number of left edge states is the number of roots, $\rho_*$ of the quadratic equation \cref{eqn:edge-poly} with modulus greater than one. The total number of edge states is twice this number by $\mathcal{PT}$-symmetry. We spend the remainder of this section counting these roots. 

In the critical case $r = 1$, the quadratic equation's solutions may be expressed without radicals,
\begin{align}
    \rho_* \in \left\{-1, \zeta_1^2 \right\},
\end{align}
so there is one left edge state if and only if $|\zeta_1| > 1$ and no edge states otherwise.

In the non-critical cases, we utilize \textit{Cohn's rules} \cite{Cohn1922} to count edge states. These rules supply conditions for a polynomial to have $q$ roots residing in the interior of the unit disk and are the essential mathematical ingredient behind the Lehmer-Schur root-finding algorithm \cite{Schur1918,Lehmer1961}. Given a polynomial over $\mathbb{C}$, $p$, its \textit{reciprocal adjoint} is the polynomial 
\begin{align*}
    p^\sharp(x) := \lim_{y \to x} y^{\deg p} p(1/y^*)^*
\end{align*}
and its \textit{Schur transform} is the polynomial
\begin{align}
    Tp := p(0)^* p - p^\sharp(0)^* p^\sharp.
\end{align}
A polynomial is \textit{self-inversive} if $|p(x)| = |p^\sharp(x)|$ for all $x \in \mathbb{C}$. Owing to the maximum modulus principle, it is straightforward to show $p$ is self-inversive if and only if $p(x) = \omega p^\sharp(x)$ for some $\omega \in \mathbb{C}$ satisfying $|\omega| = 1$. One of Cohn's rules states that if $p$ is a self-inversive polynomial, then the numbers of roots of the polynomials $p$ and $p_1(x) := \deg p * p(x) - x p'(x)$ inside the unit circle coincide. We also define the constants 
\begin{align}
    d_k := T^kp(0),
\end{align}
where $T^k$ is defined by repeated composition. Then, repeatedly applying Rouché's theorem \cite{Rouche1862} to $p(0)^* p$ and $p^\sharp(0) p^\sharp$ derives the \textit{Schur-Cohn test}, summarized as it pertains to quadratics below:
\begin{itemize}
    \item The unit disk contains none of the roots of $p$ if and only if $d_k > 0$ for all $k \in \{1, \dots, \text{deg}(p)\}$.
    \item The interior of the unit disk contains all the roots of $p$ if and only if $d_1 < 0$ and $d_k > 0$ for all $k \in \{2, \dots, \text{deg}(p)\}$.
    \item If $\deg p = 2$, $p$ has one root inside and one root outside the unit disk if $d_2 < 0$. 
\end{itemize}

In the case where $p$ is the quadratic polynomial in $\rho_*$ that defines the equation \cref{eqn:edge-poly}, $p$ is self-inversive if and only if $|\zeta_1| = 1$. Thus, when $|\zeta_1| = 1$, the number of edge states is related to the number of roots of 
\begin{align}
    f_1(\rho) = (1-\zeta_1^2)\rho - 2 \zeta_1^2 r,
\end{align}
from which it's readily verified that an edge state $(|\rho_*| > 1)$ exists if and only if $\Im(\zeta_1) > r$. 

Furthermore, we have 
\begin{align}
    d_1 &= r^2 (|\zeta_1|^4 - 1)\\
    d_2 &= (r^4 - r^2) ( |\zeta_1|^4 - 1)^2
\end{align}

Thus, applying Cohn's rules as written above in order,
\begin{itemize}
    \item two roots for $\rho_*$ have norm greater than one if and only if $d_1 > 0$ and $d_2 > 0$, equivalently, $r > 1$ and $|\zeta_1| > 1$.
    \item all roots for $\rho_*$ have norm less than one when $d_1 < 0$ and $d_2 > 0$, equivalently, $r > 1$ and $|\zeta_1| < 1$.
    \item one left edge state exists if and only if $r < 1$. 
\end{itemize}



We remark that when the asymptotic approximation of a pair of left and right edge state eigenvalues are approximately degenerate, the finite-size eigenstates may be close to a linear combination of these; this can allow for the existence of edge states without breaking $\mathcal{PT}$ symmetry, particularly in the Hermitian limit.

\subsection{Uniform chain} \label{sec:critical-EPs}

This section locates the exceptional points of a uniform open chain with defect potentials at its edges, so $t_1 = t = t_2$. 
\subsubsection{Parametrizing the EP contour}
\begin{proposition} \label{prop:EP2-parameter}
    The exceptional contour with $g > 0$ may be bijectively parametrized by the defective real eigenvalue. To be precise, given any $\epsilon \in \mathbb{R}$, there exists a unique exceptional point, $(\mu_{EP}(\epsilon), g_{EP}(\epsilon))$, for which $\epsilon$ is the unique defective eigenvalue. This point may be computed using \cref{eqn:Unique-EP-From-Evalue}
\end{proposition}
\begin{proof}
    Note $(\mu, g)$ is an exceptional point with $\epsilon$ as a defective eigenvalue if and only if $p_{\hat{H}}(\epsilon) = \partial_\epsilon p_{\hat{H}}(\epsilon) = 0$. Using the identity \cref{charPolySpecific}, these equations are equivalent to 
    \begin{align}
        \begin{pmatrix}
            U_{n-1}(\epsilon/2) & - U_{n-2}(\epsilon/2) \\
            U'_{n-1}(\epsilon/2) & - U_{n-2}'(\epsilon/2)
        \end{pmatrix} \begin{pmatrix}
            2 \mu \\
            \mu^2 + g^2
        \end{pmatrix} = \begin{pmatrix}
            U_{n}(\epsilon/2) \\
            U_{n}'(\epsilon/2)
        \end{pmatrix}. \label{eqn:Unique-EP-From-Evalue}
    \end{align}
    Let $\mathcal{W}_n(\epsilon)$ refer to the determinant of the matrix in the preceding equation, 
    \begin{align}
        \mathcal{W}_n(\epsilon) := U'_{n-1}(\epsilon/2) U_{n-2}(\epsilon/2) -U'_{n-2}(\epsilon/2) U_{n-1}(\epsilon/2),
    \end{align}
    i.e. the Wronskian of $U_{n-1}(\epsilon/2)$ and $U_{n-2}(\epsilon/2)$. Owing to the fact that the Chebyshev polynomials are orthogonal polynomials, it is well-known that this Wronskian is positive for all $\epsilon$ \cite{Karlin1960}. Furthermore, this fact may be directly verified from the three-term recurrence relation of \cref{eqn:Chebyshev-Recurrence}, which implies 
    \begin{align}
    W_n(\epsilon) = 2 \sum^{n-2}_{j = 0} U_{j}(\epsilon/2)^2
    \end{align}
    by virtue of constructing a similar recurrence relation for $W$.
    Thus, for any $\epsilon \in \mathbb{R}$, \cref{eqn:Unique-EP-From-Evalue} has a unique solution for $\mu$ and $\mu^2 + g^2$. This solution corresponds to the only exceptional point with $g > 0$ for which $\epsilon$ is a defective eigenvalue. This eigenvalue is the unique defective eigenvalue at the point $(\mu, g)$ by \cref{thm:item-one-nonsimple} of \cref{corollary}. Thus, the defective eigenvalue $\epsilon$ is in one-to-one correspondence with points on the exceptional contour.
\end{proof}

\subsubsection{Resultant analysis and large detuning limit}
Since the eigenvalue problem is exactly solvable in the case $z_1 z_n = t^2$, we assume $z_1 z_n \neq t^2$ for this subsection.

Suppose the Hamiltonian $H$ has a root $\lambda_0$ with algebraic multiplicity at least $k$. Then, it follows that 
\begin{align}
    p_H(\lambda_0) = \frac{\partial^j p_H(\lambda_0)}{\partial \lambda^j} = 0 &\quad& \forall j \in \{1, \dots, k-1\}.
\end{align}
We may find a set of model parameters for which there exists a multiple root by eliminating the variable $\lambda$ from this system of algebraic equations. For second-order exceptional points, this can be achieved utilizing resultants~\cite{Gelfand1994}. The \textit{resultant} of two monic polynomials $f(x)$ and $g(x)$ can be defined as
\begin{align}
\text{Res}_x(f, g) &:=\prod^{\deg f}_{i=1} g(x_i)
\end{align}
where $\{x_i\}$ denotes the full set of roots of $f(x)$. The resultant vanishes if and only if its inputs share one or more roots \cite{Gelfand1994}. In particular, the Hamiltonian $H$ has an eigenvalue with an algebraic multiplicity of at least two if and only if 
\begin{align} 
\text{Res}_\lambda \left(p_H(\lambda), \pdv[1]{p_H(\lambda)}{\lambda}\right) &= 0 \label{algCurve1}
\end{align} 
For the generic Hamiltonian,  \cref{algCurve1} is not satisfied for all parameters. Thus, in the generic case where only two (but not more) eigenvalues become degenerate, finding the EPs reduces to locating Hamiltonian parameters and an eigenvalue $\lambda_0$ such that 
\begin{align}
p_H(\lambda_0) &= 0=\pdv{p_H(\lambda)}{\lambda}\Bigr|_{\lambda=\lambda_0}. 
\end{align}

Computation of the general set of exceptional points reduces to finding the resultant in  \cref{algCurve1}, derived from the characteristic polynomial; We use the rescaled  polynomial $P(x) := p_{H/t}(2x)$ of \cref{charPolySpecific}. The following properties of Chebyshev polynomials are used in subsequent calculations~\cite{olver2010nist,abramowitz1972handbook,Mason1984}
\begin{align}
\frac{dU_n}{dx} &=\dfrac{(n+1) U_{n-1}(x)-nxU_n(x)}{1-x^2}, \label{eqn:cheby-der}\\
U_{n-1}\left(\frac{x+x^{-1}}{2} \right)&= \frac{x^n-x^{-n}}{x-x^{-1}}. 
\label{Joukowski}
\end{align}
The resultant of Chebyshev polynomials, calculated in~\cite{Jacobs2011,Louboutin2013}, shows that $\text{Res}\left(U_n, U_m\right) \neq 0$ if $n+1$ and $m+1$ are co-prime. Due to the denominator $(1-x^2)$ in the derivative of Chebyshev polynomials, \cref{eqn:cheby-der}, it is convenient to work with $\text{Res}_x\left[P(x), (1-x^2) dP(x)/dx\right]$ instead. To simplify this resultant, we use  the identity $(1-x^2) dP/dx=A P(x) + B U_{n-1}(x) $ where the $H$-dependent prefactors are given by 
\begin{align}
A(x)&= n \frac{\zeta_1 + \zeta_n - x (1+\zeta_1 \zeta_n)}{1-\zeta_1 \zeta_n },\label{a}\\
B(x) &= 2 \zeta_1 \zeta_n-x(\zeta_1+\zeta_n)+ (n+1)(1-\zeta_1 \zeta_n)+\frac{n(2 x-\zeta_1-\zeta_n)(2x \zeta_1 \zeta_n -\zeta_1-\zeta_n)}{1-\zeta_1 \zeta_n}.\label{b}
\end{align}

We remind the reader that the dimensionless defect potentials $\zeta_1,\zeta_n$ in \cref{a,b} are scaled by the uniform tunnelling strength $t$. Denoting the two roots of $B(x)$, \cref{b}, as $b_{\pm}$, we obtain the resultant,

\begin{align}
P(1)P(-1)\text{Res}_{x}\left(P,\frac{dP}{dx}\right)\propto
P(b_+)P(b_-).
\label{EPSurface}
\end{align}

The resultant of \cref{EPSurface} vanishes if and only if $(z_1, z_n)$ is an exceptional point. Note that \cref{EPSurface} does not yield insight for parameters $(z_1,z_n)$ that are tuned such that $P(\pm 1) = 0$ for arbitrary $n$. However, we readily identify $P(\pm 1)=0$ if and only if $(n-1)\zeta_1\zeta_n\mp n(\zeta_1+\zeta_n)+(n+1)=0$. In this case, exceptional points occur when $\pm 1$ is a double root of $P$, which occurs when 
\begin{align}
    \zeta_1 = {\zeta_n}^* \in \left\{\frac{2-2n^2 + \mathfrak{i} \sqrt{3 n^2 - 3}}{(2n-1)(n-1)}, \frac{2n^2-2+ \mathfrak{i} \sqrt{3 n^2 - 3}}{(2n-1)(n-1)} \right\}. \label{eqn:specific-ep}
\end{align}

This analytical expression also allows us to extract the asymptotic dependence of the $\mathcal{PT}$-threshold value on the detuning, where $z_1=\Delta+\mathfrak{i} \gamma=z_n^*$. Using the expansion of the roots $b_\pm$ of the quadratic expression $B(x)=0$, \cref{b}, in the limit $\Delta/t\gg 1$, and applying the method of dominant balance \cite{bender2013advanced} gives 
\begin{align}
P(b_+) P(b_-)= \frac{\Delta^2}{n^2 t^2}\left(1-\frac{1}{n}\right)^{n-2}\left[\frac{\gamma^2\Delta^{2(n-2)}}{t^{2n-2}}-1\right] + O(1,\gamma^4 \Delta^{2n-4}/t^{2n}).
\label{rhoInf}
\end{align}
It follows that the EPs determined by vanishing of  \cref{rhoInf} satisfy $\gamma_\text{EP}(\Delta)=t^{n-1}/\Delta^{(n-2)}$ in the limit $\Delta/t\gg 1$. Figure~\ref{fig:PTBreaking} shows the numerically obtained EP contours for this problem in the $(\gamma/t,\Delta/t)$ plane with varying chain sizes $n$. When $n=2$, the detuning constitutes a constant energy shift to the Hamiltonian and therefore does not change the threshold $\gamma_\mathcal{PT}(\Delta)=t$.

\subsubsection{Locating EP3s in contours of EP2s} \label{sec:EP3Contour}

In this section, we illustrate the correspondence between EP3s and cusp singularities by plotting the spectral sensitivity of eigenvalues exhibiting real-to-complex transitions.

Consider an $n \times n$ complex-valued matrix that depends polynomially on a single real or complex parameter, $x$. By applying the Newton-Puiseux theorem to the characteristic polynomial, it is well known that the eigenvalues may be locally approximated by means of a Puiseux series near some fixed $x_*$, where a Puiseux series is a generalized power series that may include terms with fractional powers of its indeterminate \cite{Canino2023}. At an EP$n$, this fractional power is generically the $n$-th root of $x$. 

Now consider an $n \times n$ complex-valued matrix, $H$, dependent upon $d$ real  parameters, $\{x_0, \dots, x_{d-1} \} \in \mathbb{R}^d$. Suppose $x_* \in \mathbb{R}^d$ is an EP2 of $H$ and let  $\lambda_*$ be a defective eigenvalue of $H(x_*)$.
Suppose $x$ is perturbed away from $x_*$ in the unit direction $v \in \mathbb{R}^d$ (i.e. $||v||_2 = 1$) by an amount $\alpha \in \mathbb{R}$, 
\begin{align}
    x = x_* + \alpha v.
\end{align}
Then, a pair of eigenvalues $\lambda$ are perturbed away from $\lambda_*$ and may be locally approximated by a Puiseux series. The leading order coefficient of this series may be determined by considering a Taylor series expansion of the characteristic equation,
\begin{align}
     0 &= p_{H(x)}(\lambda) \nonumber \\ 
     &= \frac{1}{2} \frac{\partial^2 p_{H(x_*)}(\lambda_*)}{\partial \lambda^2} (\lambda - \lambda_*)^2 + \alpha \vec{\nabla}_{x} p_{H(x_*)}(\lambda_*) \cdot v + \mathcal{O}(\alpha^2 + \alpha(\lambda - \lambda_*) + (\lambda- \lambda_*)^3),
\end{align}
where we remark the linear term in $\lambda - \lambda_*$ vanishes since $x_*$ is an EP2. 
Consequently, 
\begin{align}
    \lambda - \lambda_* = \tau \alpha^{1/2} + \mathcal{O}(\alpha),
\end{align}
where 
\begin{align}
    \tau^2 &= - 2 \left(\frac{\partial^2 p_{H(x_*)}(\lambda_*)}{\partial \lambda^2} \right)^{-1}  \vec{n} \cdot v, \\
    \vec{n} &:= \vec{\nabla}_{x} p_{H(x_*)}(\lambda_*).
\end{align}
The vector $\vec{n}$ is normal to the algebraic variety of EPs, so $\tau = 0$ if and only if $v$ is tangent to the algebraic variety of EPs.

Consider the specific $n \times n$ matrix $\hat{H}$ that is the critical non-Hermitian open chain investigated above, so $r = 1$ and $r_L = r_R = 0$, and suppose the parameters $(\mu, g)$ are perturbed in the direction $v \in \mathbb{R}^2$ away from an EP2, $(\mu_{EP}, g_{EP})$. Owing to the parametrization of proposition~\ref{prop:EP2-parameter}, there exists a unique pair of eigenvalues, $\lambda_{\pm}$, exhibiting a real-to-complex transition during this perturbation. Instead, if $v$ is tangent to the EP contour, the entire spectrum is real and all eigenvalues admit traditional power series expansions. 

If $(\mu_{EP},g_{EP})$ is an EP3, then the eigenvalues admit a Puiseux expansion with cube roots. Consequently, considering the curve to be parametrized by the defective eigenvalue $\epsilon$, if $\epsilon$ is tuned towards an eigenvalue with algebraic multiplicity three and $v$ is not tangent to the curve, the $\epsilon$-dependent coefficient $\tau$ must diverge; this divergence represents a transition from a square root to a cube root expansion.

\begin{figure}[htp!]
    \centering
    \includegraphics[width=0.9\linewidth]{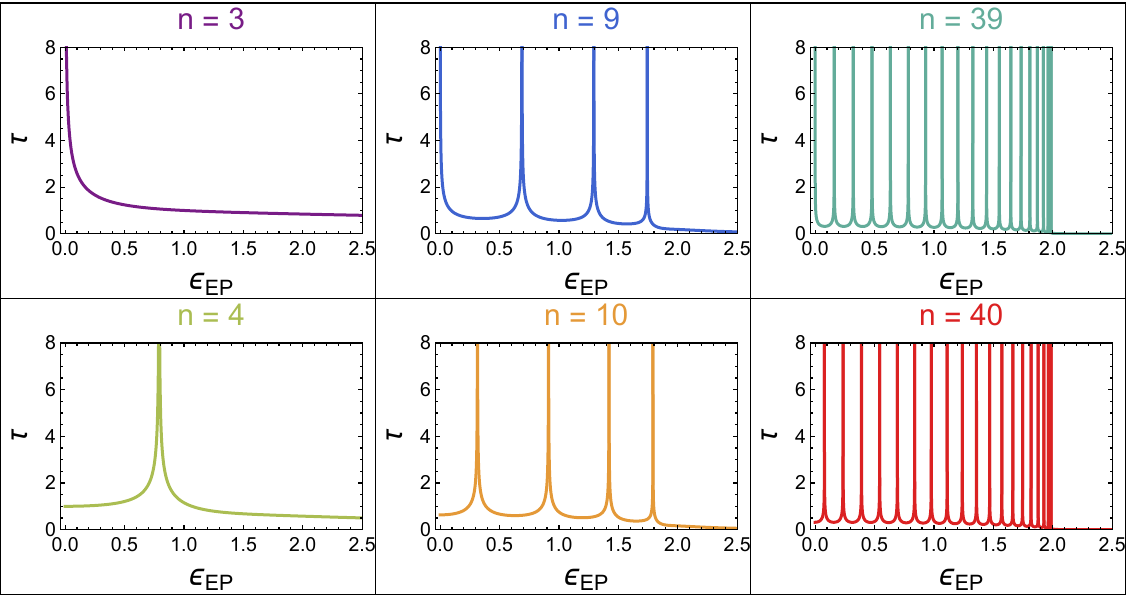}
    \caption{Coefficient $\tau$ of the square-root term in the Puiseux expansion of the defective eigenvalue, $\epsilon_{EP}$, at an EP2. The EP contour is parametrized by this eigenvalue, as explained in proposition~\ref{prop:EP2-parameter}. The direction of perturbation is in the unit direction $v = \vec{n}/||\vec{n}||$. The divergences in $\tau$ occur at EP3s,  where the eigenvalue expansion of $\epsilon_{EP}$ is expressed through cube-roots scaling faster than any square-root. Note the zero mode corresponds to an EP3 for odd lattice sizes \cite{Korff2008,jin2009solutions}. }
    \label{fig:tau}
\end{figure}

The asymptotic decay of $\tau$ as $n \to \infty$ when $\epsilon_{EP} > 2$ may be understood using the asymptotic analysis of the preceding section. Note that $\epsilon_{EP} = 2$ corresponds to the exceptional point given in \cref{eqn:specific-ep}, where $\mu = 1 + \mathcal{O}(1/n)$ and $g = \mathcal{O}(1/n)$. Furthermore, note that as $\epsilon$ and $\mu$ increase above this point, $g$ decreases exponentially in $n$. Consequently, the Hamiltonian corresponding to this exceptional point has a comparatively small anti-Hermitian part scaling with $\gamma$. This Hamiltonian can be considered "almost Hermitian", so typical Rayleigh-Schr\"odinger perturbation theory shouldn't be expected to fail as badly as in more "genuine" non-Hermitian cases, intuitively leading to a smaller coefficient $\tau$.



\section{Singular points in exceptional surfaces} \label{sec:singular-general}

Given a matrix with polynomial dependence on two parameters that has a simple spectrum almost everywhere, the set of parameters for which a non-simple eigenvalue exists is an algebraic curve.
Algebraic curves may contain a finite number of \textit{singular points}, where unique two-directional tangents may not exist. This section perturbatively examines the significance of singular points of algebraic curves of exceptional points. We assume the characteristic polynomial has real coefficients, which is necessarily the case for $\mathcal{PT}$-symmetric matrices \cite{bender2010pt}, and make two non-degeneracy assumptions in \cref{eqn:Ass-1,eqn:Ass-2}. Then, we demonstrate that EP3s generically occur at cusp singularities of EP2 curves. We conclude by presenting an example where the EP curve has an acnode and a crunode corresponding to an EP2 with multiple pairs of coalescent eigenvalues. Since the argument uses only the characteristic polynomial and not the Jordan structure of the corresponding matrix, it also applies to nondefective eigenvalue degeneracies. Thus, in this section, we use the term EP in Kato’s original broader sense \cite{Kato1995}; we refer to both EPs in the modern sense and semisimple eigenvalue degeneracies commonly termed diabolical points \cite{diabolical,diabolical2003} as EPs.

\subsection{Cusps}
Consider a Hamiltonian that is polynomial in $d \geq 2$ real or complex parameters, $H: \mathbb{F}^d \rightarrow \text{End}(\mathbb{C}^n)$ with $\mathbb{F} = \mathbb{R}$ or $\mathbb{C}$, and let $x_0 \in \mathbb{F}^d$ be an EP3. At $x_0$, there exists at least one root, $\lambda_0$, whose algebraic multiplicity increases by 2. Then, the characteristic polynomial may be written
\begin{align}
p_{H(x)}(\lambda) = \sum_{j=0}^n q_j(x) (\lambda - \lambda_0)^j,
\label{puiseux}
\end{align}
where the polynomials $q_j(x)$ satisfy 
\begin{align}
    q_j(x_0)=0\,\forall j<3 &\quad& q_3(x_0)\neq 0. \label{eqn:cubic-poly-constraint}
\end{align} 
As we perturb $x_0 \rightarrow x_0+\delta x$, we may approximate the leading-order correction $\delta\lambda=\lambda -\lambda_0$ to the eigenvalue $\lambda_0$ by neglecting $\mathcal{O}(\delta \lambda^4)$ terms from the characteristic equation,
\begin{align}
q_3(x_0)\delta\lambda^3+ q_2(x) \delta\lambda^2+ q_1(x) \delta\lambda+ q_0(x) = \mathcal{O}(\delta \lambda^4 + ||\delta x|| \delta \lambda^3).
\label{cubic}
\end{align}
In the presence of a perturbation $\delta x$, the eigenvalue $\lambda_0$ of multiplicity three splits into several eigenvalues each corresponding to the different roots $\delta \lambda$ of the cubic \cref{cubic}. Generically, this results in three distinct eigenvalues; however, if $\delta x$ is tuned so that this cubic has less than three distinct roots, then the corresponding $x$ is a leading-order approximation to an EP. These approximate EPs occur at the roots of the cubic's discriminant,
\begin{align}
Disc(x) := \tilde{q}_1^2 \tilde{q}_2^2 - 4 \tilde{q}_0 \tilde{q}_2^3 - 4 \tilde{q}_1^3 + 18 \tilde{q}_0 \tilde{q}_1 \tilde{q}_2 - 27 \tilde{q}_0^2 = 0, \label{disc}
\end{align}
where we define $\tilde{q}_j(x) := q_j(x) / q_3(x_0)$. The zero set of $Disc$ defines an algebraic variety. We note $x = x_0$ is a singular point of this variety because \cref{eqn:cubic-poly-constraint} implies $Disc$ has vanishing derivatives with respect to every coordinate $x_j$ at $x = x_0$.

Consider now the case of $d = 2$ real parameters $(\mathbb{F} = \mathbb{R})$ and assume the characteristic polynomial has real coefficients, so that the set of EPs is a real algebraic curve. We seek to approximately characterize this curve near the point $x_0$. 

Since the only term in the polynomial $Disc$ that is quadratic in $\delta x$ is the $\tilde{q}_0^2$ term, one may be inclined to argue that the EP curve is locally approximated by the equation $\tilde{q}_0 = 0$. This scenario would be facilitated by linearly approximating $\tilde{q}_0$, 
\begin{align}
    \tilde{q}_0(x_0 + \delta x) \approx \vec{\nabla} \tilde{q}_0 (x_0)\cdot \delta x.
\end{align}
However, since $x_0$ is a singular point, there is no guarantee that the EP curve admits a linear approximation. 
Instead, the Newton-Puiseux theorem guarantees that the solution to an algebraic equation in two variables may be locally approximated by setting one of the variables to be a suitably chosen Puiseux series in the other variable. To find this Puiseux series, we start by changing to coordinates adapted to the naive local approximation referred to above,
\begin{align}
    u &:= \frac{\vec{\nabla} \tilde{q}_0(x_0)}{||\vec{\nabla} \tilde{q}_0(x_0)||}, \quad\ v \perp u, \quad\ ||v|| = 1 \\
    \delta x &:= \theta u + \phi v &\quad& \theta, \phi \in \mathbb{R},
\end{align}
where we assume
\begin{align}
    \text{Assumption 1: }&\quad& \vec{\nabla} \tilde{q}_0(x_0) \neq 0 \label{eqn:Ass-1}
\end{align} 
 to be the case for the rest of this argument.
Then, we expand $Disc$ in powers of $u$ and $v$; to do this, first define two real numbers 
\begin{align}
    D_{j; \theta} := \vec{\nabla} \tilde{q}_j(x_0) \cdot u &\quad& D_{j; \phi} := \vec{\nabla} \tilde{q}_j(x_0) \cdot v &\quad& j \in \{0, 1, 2\},
\end{align}
so that the polynomials $\tilde{q}$ may be linearly approximated as
\begin{align}
    \tilde{q}_j(x) \approx \theta D_{j; \theta} + \phi D_{j; \phi}.
\end{align}
Then,
\begin{align}
    Disc(x) &= -27 \theta^2 D_{0; \theta}^2 -4 \phi^3 D_{1; \phi}^3  + K \theta \phi^2 + \mathcal{O}(\theta^3 + \phi \theta^2 +\phi^4),
\end{align}
where 
\begin{align}
    K :=& 18 \left(D_{0; \phi} D_{1; \phi} D_{2; \theta} + D_{0; \phi} D_{1; \theta} D_{2; \phi} + D_{0; \theta} D_{1; \phi} D_{2; \phi} \right)- 12 D_{1; \theta} D_{1;\phi}^2 \nonumber \\
    -& 27 D_{0; \theta} (v \cdot \vec{\nabla})^2( \tilde{q}_0)(x_0).
\end{align}

The roots $(\theta, \phi)$ of $Disc$ may be approximated with standard methods, such as the method of dominant balance \cite{bender2013advanced} or by considering a Newton polytope \cite{Canino2023}. We detail this procedure below: take the leading term of the Puiseux series for $\theta$ as a function of $\phi$, $\theta \approx A \phi^{\alpha}$ where $\alpha > 0$ and $A \neq 0$. Then, the discriminant equation implies 
\begin{align}
    -27 D_{0; \theta}^2 A^2 \phi^{2 \alpha} - 4 \phi^3 D^3_{1; \phi} + K A \phi^{2 + \alpha} \approx 0.
\end{align}
For this expression to hold at leading order, a pair of terms need to cancel while the third scales as a higher power of $\phi$. When 
\begin{align}
    \text{Assumption 2:} &\quad& D_{1; \phi} \neq 0 \label{eqn:Ass-2}
\end{align}
holds, this can only happen for the first pair with the choice $\alpha = 3/2$; thus, the exceptional contour may be approximated as the cuspoidal cubic,
\begin{align}
    Disc(x) &\approx -27 \theta^2 D_{0; \theta}^2 -4 \phi^3 D_{1; \phi}^3
\end{align}
in the vicinity of an EP3. 

We now argue that the EPs on the cuspoidal contour near an EP3 are EP2s. For a root of the cubic \cref{cubic} to have multiplicity three, the cubic's discriminant must vanish and $\tilde{q}_2^2 = 3 \tilde{q}_1$ must hold. Owing to \cref{eqn:Ass-2}, near $x_0$, we have $\tilde{q}_1 = D_{1, \phi} \phi + \mathcal{O}(\phi^{3/2})$ and $\tilde{q}_2^2 = \mathcal{O}(\phi^2)$, meaning $\tilde{q}_2^2 \neq 3 \tilde{q}_1$. Thus, $\lambda_0 + \delta \lambda$ has multiplicity two when $x \neq x_0$ is an EP sufficiently close to $x_0$.

Next, we summarize how the eigenvalue $\lambda$ is perturbed away from $\lambda_0$ along a fixed-direction perturbation $\delta x$. To leading order in $||\delta x||$, applying dominant balance to \cref{cubic} implies
\begin{align}
\delta\lambda \in \begin{cases}
\left\{ e^{2 \pi i j/3}\left(-\vec{\nabla} \tilde{q}_0(x_0) \cdot \delta x\right)^{1/3} + \mathcal{O}(||\delta x||^{2/3}) \,|\, j \in \{0, 1, 2\}\right \} \, &\text{if } u \cdot \delta x \neq 0 \\
\left\{ \pm \left(-\vec{\nabla} \tilde{q}_1(x_0) \cdot \delta x \right)^{1/2} + \mathcal{O}(||\delta x||) \right\} \cup \{\mathcal{O}(||\delta x||)\} \, &\text{if }  u \cdot \delta x = 0,
\end{cases}
\end{align}
where we note assumption 2 of \cref{eqn:Ass-2} is implicitly required for the second case. 
There is one possible perturbation for each branch of the relevant root. Notably, the order of the Puiseux series expansion decreases if the $\delta x$ is parallel to the one-directional tangent of the exceptional surface at $x_0$.

We summarize by remarking on the importance of the assumption of \cref{eqn:Ass-1}. For example, the EPs of two-parameter Hamiltonian discussed in \cite[ex. 2.10]{Barnett2023PhD},
\begin{align}
    H(x) = \begin{pmatrix}
        0 & x_1 & 0 \\
        x_2 & 0 & x_1 \\
        0 & x_2 & 0
    \end{pmatrix}
\end{align}
are always third-order; these are the points in the curve $x_1 x_2 = 0$ excluding the origin, $x_1 = x_2 = 0$. The reason why these EP3s do not exist at cusp points of EP2 curves is because Assumption 1, \cref{eqn:Ass-1}, fails.

\subsection{Acnodes and crunodes}
Here, we present an example where the exceptional contour has an acnode and a crunode.  Figure~\ref{fig:exactnm} shows the exceptional points for an open $n=5$ site chain with defects at sites $m=2$ and $\bar{m}=4$. 

\begin{figure}[!ht]
\centering
\includegraphics[width = 80mm]{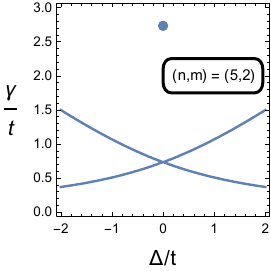}
\caption{Exceptional points (blue) of $H$ with $t_i = t$, $t_L = t_R = 0$, and $(n,m) = (5,2)$ as a function of the gain-loss strength, $\gamma/t$, and detuning, $\Delta/t$.  At small $\gamma/t$,  all five eigenvalues are real.  The points $z_m=0+\mathfrak{i}\gamma= \mathfrak{i}(\sqrt{3} \pm 1)t$ are a crunode $(-)$ and an acnode $(+)$, respectively,  with eigenvalues $\sigma(H/t) = \{0, -(\mp 1)^{1/2}3^{1/4}, (\mp 1)^{1/2} 3^{1/4}\}$. The nonzero eigenvalues at $z_m = \mathfrak{i}(\sqrt{3} \pm 1) t$ have algebraic multiplicity 2 and geometric multiplicity 1. These points are EP2s where two different pairs of eigenvalues coalesce.}
\label{fig:exactnm}
\end{figure}

\section{Conclusion}
Non-Hermitian, tridiagonal, finite-dimensional matrices with $\mathcal{PT}$-symmetry are particularly amenable to analytical treatment. They also model a vast variety of physically realizable classical and quantum systems with balanced gain and loss. Introducing just one pair of gain-loss defect potentials breaks translational symmetry in such models and makes them non-amenable to traditional, Fourier-space band-structure methods. However, experimentally implementing $\mathcal{O}(n)$ balanced gain-loss pairs in an $n$-site chain is exceptionally challenging. Therefore, we have considered models with {\it minimal non-Hermiticity} that leads to $\mathcal{PT}$-symmetry, i.e. one pair of defect potentials at mirror-symmetric sites. 

Our results include the explicit analytical expressions for a family of intertwiners, construction of an equivalent Hermitian Hamiltonian in the $\mathcal{PT}$-unbroken phase, and analytical results for the EP contours in a uniform chain with detuned defects at the end points. We have shown that higher-order EPs exist at cusp points of algebraic curves of lower-order EPs. The cases with a closed-form solution to the eigenvalue problem can potentially serve as test matrices for algorithms designed to analyze more complicated non-Hermitian matrices. Taken together, these results deepen our understanding of exceptional-point degeneracies in physically realizable models. 


Due to the well-known similarity transformation between generic tridiagonal matrices and their transpose-symmetric counterparts  \cite{santra2002non,JoglekarSaxena}, the results of this report readily generalize to tridiagonal matrices that are not transpose-symmetric, including cases with non-reciprocal couplings such as the Hatano-Nelson model \cite{Hatano1996}. 


We remark that the correspondence between the existence of edge states and $\mathcal{PT}$-symmetry breaking holds because we chose the defect location to be at the edges. If they are moved inwards, there may exist edge states with real energies \cite{Yuce2015}. 

\section*{Acknowledgment}
J.L.B. appreciates Dax Koh's proofreading of this manuscript. A majority of the results here were submitted as part of J.L.B's Ph.D. thesis \cite{Barnett2023PhD}. J.L.B. is supported by the Basque government through the BERC 2025-2028 program.

This work was supported by ONR Grant No. N00014-25-1-2160 (Y.N.J.). This research was supported in part by Perimeter Institute for Theoretical Physics (J.L.B. and Y.N.J.). Research at Perimeter Institute is supported by the Government of Canada through the Department of Innovation, Science and Economic Development and by the Province of Ontario through the Ministry of Research, Innovation and Science. 

\section*{AI Usage Statement}

J.L.B. used ChatGPT (OpenAI) to assist with writing, checking algebraic calculations, and generating Mathematica code for the figures. ChatGPT proposed an algebraic method for computing \cref{fig:tau} that was more efficient than J.L.B.'s original approach using Mathematica's AsymptoticSolve function. This led to the identification of a failure of global injectivity in an earlier parametrization by $\arg(z_1)$, presented in \cite[Fig. 3.6]{Barnett2023PhD}. ChatGPT subsequently proposed the defective-eigenvalue parametrization in Proposition 4 and supplied its proof. In response to a question about defectiveness when $t_L\ne0$, ChatGPT also supplied the result and proof presented in \cref{prop:geometric-multiplicity}. The authors reviewed and verified these contributions and take full responsibility for the manuscript.

\printbibliography
\appendix

\section{Structural Properties of the Hamiltonian}

\subsection{Symmetries} \label{app:symmetry}
We briefly summarize standard arguments regarding why the constraints $t_j > 0$ and $\gamma \geq 0$ may be imposed without loss of generality.

First, note $t_j \neq 0$, since otherwise $H$ may be treated as a direct sum of matrices with $t_j \neq 0$. We proceed to construct a diagonal similarity transformation from $H$ to a Hamiltonian with $t_j > 0$. To do this, let 
\begin{align} 
\varepsilon_j:=t_j/|t_j|\in\{\pm1\},
\qquad
d_1:=1,\qquad
d_j:=\prod_{\ell=1}^{j-1}\varepsilon_\ell
\quad \forall j\in \{2,\ldots,n\}.
\end{align}
Then, letting $D_{jk} = d_j \delta^j_k$, we find 
\begin{align}
    (D^{-1} H D)_{jk} = z_i \delta^i_{j} + |t_j| \delta^i_{j+1} + |t_i| \delta^j_{i+1} + d_n \delta^i_{\bar{j}}(t_L \delta^i_{1} + t_R \delta^i_{n}).
\end{align}

If $\gamma < 0$, then conjugating the Hamiltonian with $\mathcal{P}$ results in a similar Hamiltonian with $\gamma > 0$, so we may assume $\gamma > 0$ without loss of generality. Therefore, the site $m$ is the gain site and its mirror-symmetric site $\overline{m}$ is the loss site. 

\subsection{Geometric Multiplicity} \label{app:geom-mult}

Below, we explain how when $t_L \neq t_R$, eigenvalues of $H$ as defined in \cref{TriDiag} are never degenerate.

\begin{proposition}
\label{prop:geometric-multiplicity}
Let $H$ be the matrix defined in \cref{TriDiag},
with $t_j\neq 0$ for all $j\in\{1,\ldots,n-1\}$ and
$t_L \neq t_R$ or $t_L = t_R = 0$. Then every eigenvalue of $H$ has geometric
multiplicity one. Consequently, an eigenvalue of algebraic
multiplicity $m$ corresponds to a single Jordan block of
size $m$.
\end{proposition}

\begin{proof}
First, suppose $t_L\neq t_R$. Let $u,v\in\mathbb C^n \setminus \{0\}$
be eigenvectors corresponding to the same eigenvalue
$\lambda$. Since $Hu=\lambda u$ and $Hv=\lambda v$,
\begin{align}
0
&=u^T Hv-(Hu)^T v \nonumber\\
&=u^T(H-H^T)v \nonumber\\
&=(t_L-t_R)(u_1v_n-u_nv_1).
\end{align}
As $t_L-t_R\neq 0$, it follows that
\begin{equation}
\det \begin{pmatrix}
    u_1 & v_1 \\
    u_n & v_n
\end{pmatrix} =0.
\end{equation}
Thus, there exists a nonzero vector $a \in \mathbb{C}^2$ in the kernel of this $2 \times 2$ matrix. Define
\begin{align}
    w := a_1 u + a_2 v.
\end{align}
Then, $w$ is an eigenvector of $H$ satisfying $w_1=w_n=0$.

Since $Hw=\lambda w$, the eigenvalue equation at the
first site gives $t_1w_2=0$, hence $w_2=0$.
Applying the eigenvalue equation successively at
sites $2,\ldots,n-1$ yields $w_j=0$ for every $j$.
Therefore $w=0$, proving that $u$ and $v$ are
linearly dependent.

If $t_L=t_R=0$, the same conclusion was shown to apply in \cite[Thm. 4.2]{elliott1953characteristic}. Essentially, the argument from the preceding argument continues to apply: the eigenvalue equation at the first
site determines $w_2$ from $w_1$, and the remaining
components are determined recursively by the
nonvanishing interior couplings. Hence an eigenvector
is uniquely determined by its first component, meaning every eigenspace is one-dimensional.

Finally, since every eigenspace is one-dimensional,
the Jordan canonical form contains exactly one
Jordan block for each distinct eigenvalue. Its size
equals the algebraic multiplicity of that eigenvalue.
\end{proof}

\section{An alternative proof of maximal symmetry breaking}
\label{app:maximal-breaking}

In this appendix, we summarize an argument that has been discussed in \cite{MyFirstPaper,Barnett2023PhD} that establishes maximal symmetry breaking.

Suppose there exists a real eigenvalue of $H$, $\lambda \in \mathbbm{R} \cap \sigma(H)$. Let $\ket{\psi} = \sum_{k = 1}^{2m} \psi_k e_k$ be an eigenvector of $H$ corresponding to the eigenvalue $\lambda$. We remark that $\psi_1 \neq 0$ must hold since otherwise the eigenvalue equations would contradictorily imply $\psi = 0$; thus, we may rescale the eigenstate $\psi \to \psi/\psi_1$. The eigenvalue equations thus imply that the eigenstate components can be taken to be real for all sites on the left half of the lattice, $\psi_k \in \mathbbm{R}\,\forall k \leq m$.
Since $H$ is $\mathcal{PT}$-symmetric and the geometric multiplicity of $\lambda$ is one \cite[Thm. 4.2]{elliott1953characteristic}, every eigenvector corresponding to $\lambda$, $|\psi\rangle=\sum_{k=1}^{2m} \psi_k e_k \in \ker(\lambda - H) \setminus \{0\}$, is an eigenvector of the antiunitarity operator $\mathcal{PT}$ \cite{bender1999pt}. Since $\mathcal{PT}$ is antiunitary, its eigenvalues necessarily reside on the unit circle \cite{Wigner1960}; thus, there exists a phase $\chi \in [0, 2 \pi)$ such that $\psi_{\bar{k}} e^{\mathfrak{i} \chi} = \psi_k \,\forall k \leq m$.
With these observations in mind, the eigenvalue equations at the nearest-neighbour defect sites  $(m,m+1)$ are equivalent to
\begin{equation}
\begin{pmatrix}
(z_m - \lambda) \psi_m + t_{m-1} \psi_{m-1} & t_m \psi_m \\
t_m \psi_m & (z_{m+1} - \lambda) \psi_m + t_{m-1} \psi_{m-1}
\end{pmatrix}
\begin{pmatrix}
1 \\ e^{\mathfrak{i} \chi}
\end{pmatrix} =0.
\end{equation}
For this matrix to have a nontrivial kernel, its determinant must vanish.  However, if $\gamma > t_m$, the determinant is strictly positive. The contradictory assumption was taking $\lambda \in \mathbbm{R} \cap \sigma(H)$, thus, every eigenvalue of $H$ has a nonzero imaginary part when $\gamma > t_m$.

\section{Alternating gain and loss}
\label{app:alternating}

We denote the Hamiltonian of an open SSH chain with $n=2k$ sites, $k\geq2$, by $H_{\mathrm{SSH}}$ and define a non-Hermitian perturbation, $H_{\mathrm{alt}}$ by applying balanced gain and loss uniformly to the two sublattices; explicitly, 
\begin{align}
 (H_{\mathrm{SSH}})_{jk}
 &= t_k \delta^j_{k+1} + t_j \delta^k_{j+1},
 &t_{2j-1}&=t_1,\qquad t_{2j}=t_2,\\
 \chi_{jk} &= (-1)^j \delta_{jk} ,
 & \\
 H_{\mathrm{alt}}&=H_{\mathrm{SSH}}+i\gamma\chi ,
\end{align}
where $\chi$ is an operator that realizes particle-hole symmetry in $H_{\mathrm{SSH}}$ owing to the relation
\begin{align}
\chi^{-1} H_{\mathrm{SSH}} \chi = -H_{\mathrm{SSH}} &\quad& \chi^2 = \chi \chi^\dag = \mathbb{1}_n.
\end{align}
Particle-hole symmetry implies that the spectrum of $H_{\mathrm{SSH}}$ is symmetric with respect to a reflection through the origin. To simplify analysis, we rescale the Hamiltonians and parameters by a factor of $t_2^{-1}$, $\hat{H}_{\mathrm{SSH}, \mathrm{alt}} = H_{\mathrm{SSH}, \mathrm{alt}}/t_2$, as in \cref{eqn:rescale-variables}. 

We proceed to relate the ground state energy of $H_{\mathrm{SSH}}^2$ to the $\mathcal{PT}$-threshold of $H_{\mathrm{alt}}$. A straightforward argument utilizing the particle-hole symmetry of $H_{\mathrm{SSH}}$ and the Freshman's dream may be used to deduce the spectrum of $H_{\mathrm{alt}}$ \cite{MartinezMartinez2024,BarnettElGanainy2026} 
\begin{align}
 \sigma(H_{\mathrm{alt}})
 =\left\{\pm\sqrt{\epsilon -\gamma^2}\,|\, \epsilon \in \sigma(H_{\mathrm{SSH}}^2)\right\}.
\end{align}
Thus, 
\begin{align}
\gamma_{\mathrm{PT}}^2= \min \sigma(H_{\mathrm{SSH}}^2). \label{eq:alt-threshold}
\end{align}

Hence, eigenvalue inclusion results for $H_{\mathrm{SSH}}$ may be used to bound $\gamma_{\mathrm{PT}}$. We proceed to place an upper bound on the smallest positive eigenvalue of $H_{\mathrm{SSH}}$ in the following lemma:

\begin{lemma}
\label{lem:alt-variational}
For $k\geq2$ and $r>0$, there exists an eigenvalue, $\epsilon$, of
$\hat{H}_{\mathrm{SSH}}$ satisfying
\begin{equation}
 0 < \epsilon < r^k \sqrt{\frac{1-r^2}{1-r^{2k}}},
 \label{eq:alt-variational}
\end{equation}
where the quantity $\frac{1-r^2}{1-r^{2k}}$ is interpreted by continuity as $1/k$ at $r = 1$.
\end{lemma}

\begin{proof}
Define a vector supported on the sites labeled by odd indices,
\begin{align}
 \psi_{2 j+1} &= (-r)^j, &\quad& 
 \psi_{2j } = 0.
\end{align}
Then, we may deduce an upper bound to the ground state energy of $H_{\mathrm{SSH}}^2$ by applying the Rayleigh--Ritz variational principle \cite{Ritz1908}. Denoting this energy by $\epsilon^2$,
\begin{align}
    \epsilon^2 \leq \frac{\braket{\psi|\hat{H}^2_{\mathrm{SSH}} \psi}}{\braket{\psi|\psi}} = \frac{r^{n}-r^{2+n}}{1-r^{n}} 
\end{align}

We now detail why the upper bound of \cref{eq:alt-variational} is strict. Suppose on the contrary that $\lambda := \frac{r^n - r^{2+n}}{1-r^n}$ was the smallest eigenvalue of $\hat{H}_{\mathrm{SSH}}^2$. Then, the equation $\lambda \braket{\psi| \psi} = \braket{\psi|\hat{H}_{\mathrm{SSH}}^2 \psi}$ would necessarily imply that $\hat{H}_{\mathrm{SSH}}^2 \psi = \lambda \psi$ owing to the spectral decomposition of Hermitian matrices. However, $\hat{H}_{\mathrm{SSH}}^2 \psi$ only has nonzero support on the site $2k-1$ while $\lambda \psi$ is supported on every site with an odd index. Thus, $\hat{H}_{\mathrm{SSH}}^2 \psi \neq \lambda \psi$ whenever $k > 1$; the contradiction was assuming $\lambda$ was the smallest eigenvalue of $\hat{H}_{\mathrm{SSH}}^2$.  
\end{proof}

We combine this lemma with \cref{inclusionTheorem} to bound $\gamma_{\mathcal{PT}}$ in the following proposition.
\begin{proposition}
\label{thm:alt-inclusion}
Let $n=2k\geq4$ and $r_{\mathrm c}=k/(k+1)$.
The $\mathcal{PT}$-threshold of the open SSH chain with alternating
gain and loss satisfies
\begin{equation}
\begin{cases}
 0<\gamma_{\mathrm{PT}}<|t_2-t_1|,
       &0<r<r_{\mathrm c},\\[2pt]
 \gamma_{\mathrm{PT}}=|t_2-t_1|,
       &r=r_{\mathrm c},\\[6pt]
 \gamma_{\mathrm{PT}}>|t_2-t_1|,
       &r>r_{\mathrm c}.
\end{cases}
\label{eq:alt-inclusion}
\end{equation}
Moreover, for every $r>0$,
\begin{equation}
 \gamma_{\mathrm{PT}}
 <t_2\min\left\{\xi_{k-1}, r^k \sqrt{\frac{1-r^2}{1-r^{2k}}}\right\}.
 \label{eq:alt-combined}
\end{equation}
\end{proposition}

\begin{proof}
These inequalities are derived by bounding the ground state energy of $H_\mathrm{SSH}^2$ and using \cref{eq:alt-threshold}. The bound 
\begin{align}
    \gamma_\mathcal{PT} < t_2 r^k \sqrt{\frac{1-r^2}{1-r^{2k}}}
\end{align}
is a restatement of lemma~\ref{lem:alt-variational}.
The remaining bounds are derived from \cref{inclusionTheorem} with the choice $\mu = \beta = 0$. We detail how each of these bounds is derived by dividing the analysis into different parameter regimes:

\begin{itemize}
    \item Suppose $r > r_c$. Then, since \cref{ineq-1,ineq-2} are satisfied, every eigenvalue of $\hat{H}_{\mathrm{SSH}}$ resides in the bands $B_\pm$. Thus, $\gamma_\mathcal{PT} > |t_2 - t_1| = \xi_k$. Furthermore, since \cref{ineq-2} is satisfied and there is an eigenvalue in the interval $\mathcal{I}_+(k-1)$, it follows that $\gamma_\mathcal{PT} < t_2 \xi_{k-1}$.
    \item If $r = r_c$, direct inspection of the characteristic polynomial implies $|t_2 - t_1|$ is an eigenvalue of $H_{\mathrm{SSH}}$. The remaining $k-1$ positive eigenvalues reside in the bands since \cref{ineq-1} is satisfied; thus, $\gamma_\mathcal{PT} = |t_2 - t_1|$.
    \item Suppose $r < r_c$. Then, there exists an eigenvalue of $\hat{H}_{\mathrm{SSH}}$ residing in the band gap since \cref{eqn:gap-eval-exists} is satisfied. Thus, $\gamma_\mathcal{PT} < |t_2 - t_1|$.
\end{itemize}

Lastly, note $\gamma_\mathcal{PT} > 0$ since $0$ is not an eigenvalue of the finite-size open SSH chain.
\end{proof}

We conclude from the inequalities of the preceding proposition that the $\mathcal{PT}$-unbroken phase is fragile in the non-trivial phase and robust in the trivial phase; numerical evidence for this claim was presented in \cite{Klett2017}. 

We remark that the $\mathcal{PT}$-threshold at $r = r_c$, $\gamma = |t_2 - t_1|$, is also the $\mathcal{PT}$-threshold when the Hamiltonian $H_{\mathrm{alt}}$ is supplanted with periodic boundary conditions \cite{Schomerus2013,Harter2018Fragile}.


\section{Odd non-critical chains} \label{app:odd}

Below, we remark on the complexity of the spectrum for non-uniform odd SSH chains perturbed by complex boundary potentials.
\begin{proposition}[Real eigenvalues of odd noncritical chains]
\label{prop:odd-noncritical-real-eigenvalues}
Let \(n=2k+1\), and assume 
\[
    z_1=z_n^*=\Delta+\mathfrak{i}\gamma,
    \qquad
    \gamma\neq 0,
    \qquad
    t_1\neq t_2,
    t_L+t_R\in\mathbb{R},
    \qquad
    t_Lt_R\in\mathbb{R}.
\]
Then, \(\lambda\in\sigma(H)\cap\mathbb{R}\), is a real eigenvalue if and only if it simultaneously satisfies \cref{eqn:real-evalue-Lobatto,eqn:real-evalue-odd-re-part}. 
Consequently, \(H\) has at most two distinct real eigenvalues that must reside in the set
\begin{align}
    \lambda\in
    \bigl\{
        2\Delta-(t_L+t_R),
        2\Delta+(t_L+t_R)
    \bigr\}. \label{eqn:odd-real-possibilities}
\end{align}
If \(t_L+t_R=0\), then \(H\) has at most one distinct real eigenvalue.
\end{proposition}

\begin{proof}
For \(n=2k+1\), the characteristic equation is
\begin{align}
0={}&
    (\lambda-2 \Delta)U_k(Q)-(t_L+t_R)
    \nonumber\\
&+
    \frac{
        \lambda(z_1z_n-t_Lt_R)-z_1t_1^2-z_nt_2^2
    }{t_1t_2}
    U_{k-1}(Q),
\label{eq:odd-characteristic-equation}
\end{align}
where $Q$ was defined in \cref{eqn:Q}. Suppose that \(\lambda\in\mathbb{R}\). Then \(Q\in\mathbb{R}\),
so \(U_k(Q)\) and \(U_{k-1}(Q)\) are real. 
Taking the imaginary part of
\cref{eq:odd-characteristic-equation} implies
\begin{align}
    -\gamma\frac{t_1^2-t_2^2}{t_1t_2}
    U_{k-1}\bigl(Q\bigr)=0.
\end{align}
Since \(\gamma\neq0\) and \(t_1\neq t_2\), it follows that
\begin{align}
    U_{k-1}\bigl(Q\bigr)=0. \label{eqn:real-evalue-Lobatto}
\end{align} Since $Q$ resides in the root set of $U_{k-1}$ that is explicitly presented in \cref{eqn:Chebyshev-2nd-Roots}, it follows that
\begin{align}
    U_k\bigl(Q\bigr)\in\{-1,1\}. \label{eqn:Chebyshev-sign-at-Lobatto}
\end{align} 
Consequently, the real part of
\cref{eq:odd-characteristic-equation} reduces to
\begin{align}
     \lambda-2\Delta
    =
    U_k\bigl(Q(\lambda)\bigr)(t_L+t_R). \label{eqn:real-evalue-odd-re-part}
\end{align}
Insertion of \cref{eqn:Chebyshev-sign-at-Lobatto} into \cref{eqn:real-evalue-odd-re-part} implies the set inclusion of \cref{eqn:odd-real-possibilities}.

Conversely, if \cref{eqn:real-evalue-odd-re-part,eqn:real-evalue-Lobatto} hold for some real $\lambda$, then the characteristic equation is satisfied, rendering $\lambda$ a real eigenvalue of $H$.
\end{proof}

\section{Additional eigenvalue inclusion result} \label{app:Cassini}

Now, we find a sufficient condition for the existence of non-real eigenvalues using the Brauer-Ostrowski ovals theorem \cite{Brauer1947,Ostrowski1937,Varga2004}; this subset is the region in parameter space satisfying the inequalities of \cref{unbrokenIneq1,unbrokenIneq2}. Let 
\begin{align}
    C(w_1,w_2; b) = \{w \in \mathbb{C} \,|\, |w - w_1|  \cdot |w-w_2| \leq b\}
\end{align} denote a Cassini oval. By the Brauer-Ostrowski ovals theorem, all eigenvalues of $H$ are elements of the union of Cassini ovals centered at pairs of diagonal elements of $H$. To be precise,
\begin{align}
\sigma(H) \subseteq &C(0,0;(t_1+t_2)^2) \cup C(0,z_n; (t_1+t_2)(t_1+|t_R|)) \cup \nonumber \\
&C(z_1, 0; (t_1+t_2)(t_1+|t_L|)) \cup C(z_1,z_n;(t_1+|t_L|)(t_1+|t_R|)). \label{Cassini}
\end{align}
Eigenvalues are continuous in the arguments of a continuous matrix function \cite{Kato1995}. Thus, if the union of the Cassini ovals in \cref{Cassini} contains disjoint components, then each component contains at least one eigenvalue of $H$ if that component contains one of the diagonal elements of $H$. In particular, if both of the inequalities
\begin{align} 
\left[|z_1|^2-(t_1+t_2)^2\right] \gamma  &> |z_1| (t_1+|t_L|)(t_1+|t_R|), \label{unbrokenIneq1} \\
2(t_1+t_2) &< \left(|z_1| + \sqrt{|z_1|^2 - 4t_1^2 - 4 \min\{|t_L|^2, |t_R|^2\} }\right)\label{unbrokenIneq2}
\end{align}
hold, then there exist disjoint components containing the points $z_1$ and $z_n=z_1^*$ that do not contain the real axis, implying the existence of at least two eigenvalues with nonzero imaginary parts.

This approach uses only the matrix elements of $H$. It does not exploit the detailed structure of the characteristic polynomial, which yields the sharper criteria of \cref{sec:evalue-inclusion,sec:edges}.
\end{document}